\documentclass[smallabstract,smallcaptions]{dccpaper}

\usepackage{amsthm}
\usepackage{amsmath}
\usepackage{amssymb}
\usepackage{color}
\usepackage{url}
\usepackage{epsfig}
\usepackage{citesort}
\usepackage{multirow}
\usepackage{xcolor}
\usepackage{comment}

\usepackage{appendix}
\usepackage{graphicx}
\usepackage{enumerate}
\usepackage{tikz}
\usepackage[linesnumbered,ruled,vlined]{algorithm2e}
\usepackage{algpseudocode}
\newtheorem{theorem}{Theorem}

\newtheorem{proposition}[theorem]{Proposition}

\newtheorem{remark}[theorem]{Remark}
\newtheorem{lemma}[theorem]{Lemma}

\graphicspath{ {./Figures/} }

\newlength{\figurewidth}
\newlength{\smallfigurewidth}

\begin{document}

\title
{\large
\textbf{Succinct Data Structure for Path Graphs}
}

\author{%
Girish Balakrishnan$^{\ast}$, Sankardeep Chakraborty$^{\dag}$, N S Narayanaswamy$^{\ast}$,\\ and Kunihiko Sadakane$^{\dag}$\\[0.5em]
{\small\begin{minipage}{\linewidth}\begin{center}
\begin{tabular}{ccc}
$^{\ast}$Indian Institute of Technology Madras, & \hspace*{0.5in} & $^{\dag}$University of Tokyo, \\
Chennai, India && Tokyo, Japan\\
\url{girishb@cse.iitm.ac.in} && \url{sankardeep.chakraborty@gmail.com}\\
\url{swamy@cse.iitm.ac.in} && \url{sada@mist.i.u-tokyo.ac.jp}

\end{tabular}
\end{center}\end{minipage}}
}

\maketitle
\thispagestyle{empty}

\begin{abstract}
We consider the problem of designing a succinct data structure for {\it path graphs} (which are a proper subclass of chordal graphs and a proper superclass of interval graphs) on $n$ vertices while supporting degree, adjacency, and neighborhood queries efficiently. We provide the following two solutions for this problem: 
\begin{enumerate}
    \itemsep0em 
    \item an $n \log n+o(n \log n)$-bit succinct data structure that supports adjacency query in $O(\log n)$ time, neighborhood query in $O(d \log n)$ time   and finally, degree query in $\min\{O(\log^2 n), O(d \log n)\}$ where $d$ is the degree of the queried vertex.
    \item an $O(n \log^2 n)$-bit space-efficient data structure that supports adjacency and degree queries in $O(1)$ time, and the neighborhood query in $O(d)$ time where $d$ is the degree of the queried vertex.
\end{enumerate}
Central to our data structures is the usage of the classical heavy path decomposition by Sleator and Tarjan~\cite{ST}, followed by a careful bookkeeping using an orthogonal range search data structure using wavelet trees~\cite{Makinen2007} among others, which maybe of independent interest for designing succinct data structures for other graph classes.
\end{abstract}

\section{Introduction}
An \textit{intersection graph} $G=(V,E)$ is an undirected graph whose vertices are mapped by $f$ to a family of sets $F$ such that vertex $a_1$ is adjacent to $a_2$ in $G$ if and only if $f(a_1) \cap f(a_2) \neq \phi$. Based on the family of sets $F$ we get different graph classes. For instance, if $F$ is the set of intervals on the real number line, then we get \textit{interval graphs}. Yet another example is \textit{chordal graphs}, defined as the intersection graph of sub-trees of a tree. \textit{Path graphs} is the class of graphs obtained when $F$ is the set of paths, $P_1,\ldots,P_n$ in a tree $T$ such that two paths intersect if and only if the corresponding vertices are adjacent. 
It is well-known that the class of path graphs is a proper subclass of chordal graphs and a proper superclass of interval graphs~\cite{agtpg}.

In this work, we address the problem of designing a succinct data structure for the class of path graphs so that basic navigational queries such as degree, adjacency, and neighborhood can be answered efficiently. Formally, given a set $T$ consisting of combinatorial objects with a certain property, our goal is to store any arbitrary member $x \in T$ using the information-theoretic minimum of $\log(|T|)+o(\log(|T|))$ bits (throughout this paper, $\log$ denotes the logarithm to the base $2$) while still being able to support the queries efficiently on $x$. Recently, Acan et al.~\cite{HSSS} showed that the information-theoretic lower bound for representing unlabeled interval graphs with $n$ vertices is at least $n \log n$ bits, and as path graphs are a proper superclass of interval graphs, this lower bound also holds true for path graphs. Interestingly, we manage to construct an $n \log n + o(n\log n)$-bit data structure for representing path graphs matching this information-theoretic lower bound, thus, obtaining succinct data structure for path graphs for the first time in literature. This is the main contribution of this work. We leave the question  of whether path graphs are only a constant factor larger in size than the class of interval graphs as an open problem.

\noindent
{\bf Previous Related Work.} There already exists a huge body of work on representing various classes of graphs succinctly. A partial list of such special graph classes would be trees~\cite{ChakrabortyS19,MunroR01}, planar graphs~\cite{AleardiDS08}, partial $k$-tree~\cite{FarzanK14}, and arbitrary graphs~\cite{FarzanM13}. Recent results have appeared in literature for intersection graphs like interval graphs due to Acan et al~\cite{HSSS} and chordal graphs due to Munro and Wu\cite{Munro_Wu}. For interval graphs,~\cite{HSSS} gives an $n \log n + O(n)$ bit succinct data structure that supports degree and adjacency queries in $O(1)$ time while neighborhood query in constant time per neighbour. In the case of chordal graphs, \cite{Munro_Wu} gives an $n^2/4+o(n^2)$ bit succinct  data structure that supports adjacency query in $f(n)$ time where $f(n) \in \omega(1)$, degree of a vertex in $O(1)$ time and neighborhood in $(f(n))^2$ time per neighbour. The main motivation behind our work stems from these two above-mentioned works. Since path graphs is a strict subclass of chordal graphs and a strict superclass of interval graphs it would be interesting to see whether one can design such an efficient data structure for path graphs as well.

\noindent
{\bf Our Results.} Before we get to our results, note the following terminology for graph $G=(V,E)$:
\begin{itemize}
    \itemsep0em 
    \item for $u, v \in V$, adjacency query checks if $\{u,v\} \in E$,
    \item for $u \in V$, the neighborhood query returns all the vertices that are adjacent to $u$ in $G$, and
    \item for $u \in V$, the degree query returns the number of vertices adjacent to $u$ in $G$.
\end{itemize}

Our primary result in this work is an $n \log n+  o(n\log n)$-bit succinct representation for unlabelled connected path graphs. It is obtained from the  clique tree representation $(T,P_1,\ldots,P_n)$~\cite{Gavril_path}~\cite{MonmaW86} on the input path graph.   Here $T$ is the clique tree~\cite{MonmaW86} and $P_i,1 \le i \le n$, are the paths in it. We then store $T$ succinctly along with the end-points of the paths $P_i$ in it. 
 Formally we have the following result.
\begin{theorem}
\label{thm:succthm}
Path graphs have an $n \log n + o(n\log n)$-bit succinct representation. The succinct representation constructed from the clique tree representation supports for a vertex $u$ the following  queries:
\begin{enumerate}
    \itemsep0em 
    \item adjacency query in $O(\log n)$ time,
    \item the neighborhood query in 
$O(d_u \log n)$ time, and 
    \item the degree query in  $\min\{O(\log^2 n), O(d_u \log n)\}$ time
\end{enumerate}
where $d_u$ is the degree of vertex $u$.
\end{theorem}
The central tool that we use in obtaining the above succinct data structure result and the space-efficient data structure is heavy path decomposition (HPD)~\cite{ST} performed on the clique tree $T$. The HPD when performed on the clique tree $T$ gives  the heavy path tree $\mathcal{T}$; explained in Section \ref{sec:hpd}. Each node of the heavy path tree $\mathcal{T}$ corresponds to a heavy path of the clique tree $T$. The property that heavy path tree $\mathcal{T}$ has at most $\lceil \log n \rceil$ levels helps us achieve the query times of Theorem \ref{thm:succthm}. Additionally, we observe that the intersection of the paths $P_1, \ldots, P_n$ with each heavy path  defines a natural interval graph giving us the space-efficient data structure for path graphs. Further we observe that the union of these interval graphs corresponding to nodes in the same level of the heavy path tree is also an interval graph.
To obtain the space-efficient data structure, we store the interval graphs at each level of the heavy path tree using the results from \cite{HSSS} and organize them into at most $\log n$ levels. Even though we use additional $\log n$ factor storage in the space-efficient data structure over the succinct representation, we can respond to all the queries more efficiently. This is our second result. 


\begin{theorem}
\label{thm:nlog2n}
There exists a space-efficient representation for path graphs using $O(n \log^2 n)$ bits. The representation supports the following queries for a vertex $u$:
\begin{enumerate}
    \itemsep0em 
    \item the adjacency and degree queries in $O(1)$ time, 
    \item the neighborhood query in $O(d_u)$ time where $d_u$ is the degree of the vertex $u$.
\end{enumerate}
\end{theorem}
\noindent
The increased efficiency of the space-efficient data structure comes at the expense of increased space which arises due to the duplication of edges of path graph among the $\log n$ interval graphs. Another difference is that the succinct data structure  performs orthogonal range search to implement the queries while the space-efficient data structure delegates the queries to those of the underlying interval graph as implemented in~\cite{HSSS}. 

All the preliminary terminology and concepts required for rest of the sections are  in Section~\ref{sec:prelims} and~\ref{sec:proprocessingT}. In Section~\ref{sec:succrep}, a succinct representation for path graphs is presented and Section \ref{sec:seds} describes the space-efficient data structure.

\section{Preliminaries}
\label{sec:prelims}
For a graph $G$, through out the paper we denote the set of vertices and edges by $V(G)$ and $E(G)$, respectively. Familiarity with basic graph theory as in~\cite{gtdr} and graph algorithms as given in~\cite{CLRS} is expected. 

\subsection{Path Graphs and Its Properties}

A graph $G$ is a \textit{path graph} if there exists a tree $T$ and family of paths $\mathcal{P}=\{P_1,\ldots,P_n\}$ in $T$ such that $G$ is the intersection graph of paths in $\mathcal{P}$. $G$ is said to have the representation $(T,\mathcal{P})$; see Figure~\ref{fig:pathgraph}. A vertex $a \in V(G)$ is \textit{simplicial} if the set of vertices adjacent to $a$, denoted $N(a)$, induces a complete sub-graph of $G$ \cite{agtpg}. The ordering $\rho=[a_1,\ldots,a_n]$ of $V(G)$ is called a \textit{perfect elimination scheme} if for all $i, X_i=\{ a_j \in N(a_i) \mid j > i,\}$ is complete. Every path graph has a simplicial vertex and a perfect elimination scheme. It is well known that any simplicial vertex can start a perfect elimination scheme; see Theorem 4.1 and Lemma 4.2 of~\cite{agtpg} for more details. Let $\mathcal{C}$ be the set of maximal cliques of $G$ and for every $a \in V(G)$ let $\mathcal{C}_a=\{ C| C \in \mathcal{C} \text{ and } a \in V(C)\}$. Consider a tree $T$ with $V(T)=\mathcal{C}$ such that for every $a \in V(G)$, $\mathcal{C}_a$ induces a sub-tree $T_a$ of $T$. $G$ is a choral graph if it is the intersection graph of set of such induced sub-trees. For chordal graphs such a tree $T$ is called the \textit{clique tree} of $G$~\cite{agtpg}. Clique tree can be computed in polynomial time~\cite{Gavril_path}. The following is a characterisation of path graphs as a sub-class of chordal graphs~\cite{Gavril_path}\cite{MonmaW86}.

\begin{figure}[ht]
\centering
\includegraphics[width=1\textwidth]{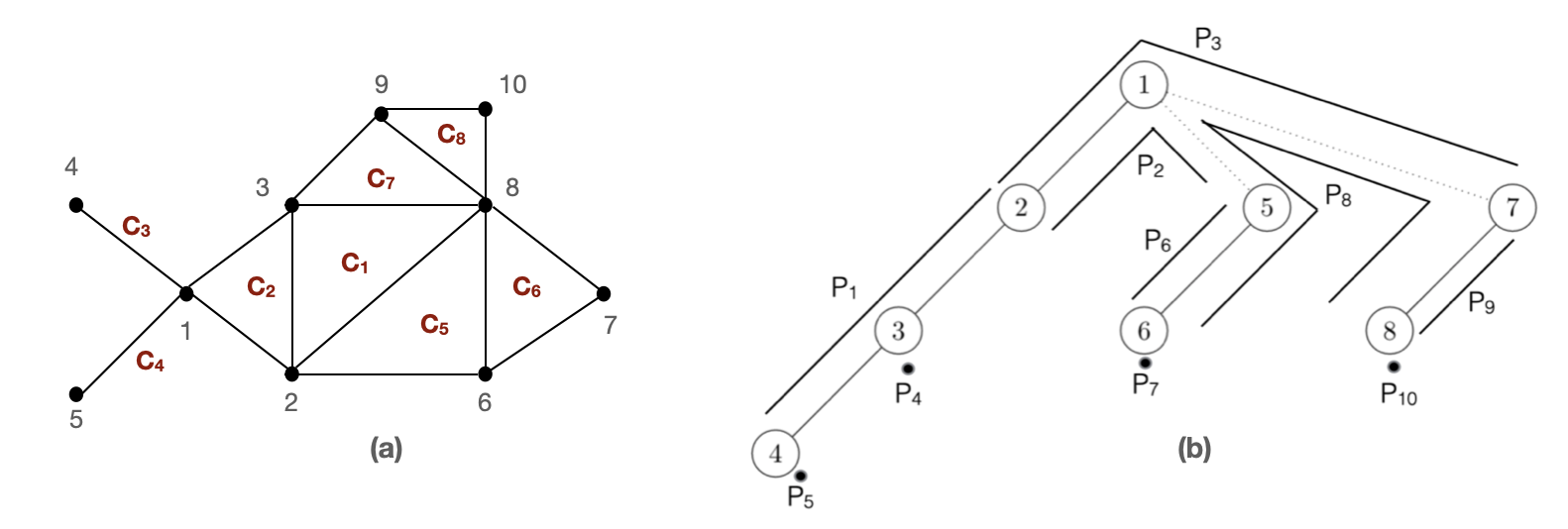}
\caption{(a) Path graph $G$ with maximal cliques $\mathcal{C}=\{C_1,\ldots,C_8\}$, (b) Clique tree representation of $G$ where each node $i$ corresponds to $C_i \in \mathcal{C}$. For vertices $u,v \in V(G)$ we have paths $P_u, P_v$ such that $V(P_u) \cap V(P_v) \ne \phi$ if and  only if $\{u,v\} \in E(G)$. The clique  tree shown  here  is pre-processed as explained in Section~\ref{sec:proprocessingT}.}
\label{fig:pathgraph}
\end{figure}

\begin{theorem}
\label{thm:pg}
The graph $G$ is a path graph if and only if there exists a clique tree $T$, such that for every $v \in V(G)$, the set of maximal cliques containing $v$ form a path in $T$.
\end{theorem}

\noindent
In this paper, we are concerned with the construction of a succinct representation for path graphs and the construction mechanism takes as input, $(T,\mathcal{P})$. Also, apart from clique tree $T$ we will introduce the heavy path tree $\mathcal{T}$ in the next section. Elements of $V(T)$ and $V(\mathcal{T})$ will be henceforth referred to as nodes of $T$ and $\mathcal{T}$, respectively whereas for any graph $G$, elements of $V(G)$ will be referred to as its vertices. The following is known from~\cite{Gavril_path}.

\begin{remark}
\label{rem:maxn}
The number of maximal cliques in a path graph $G$ with $n$ vertices is at most $n$.
\end{remark}

\subsection{Heavy Path Decomposition}
\label{sec:hpd}
Heavy path decomposition (HPD) was introduced in \cite{ST} and used in \cite{GO} and \cite{FGSVJ} for rooted trees. In the heavy path decomposition for a rooted tree $T$, each internal node $u$ selects an edge $(u,v)$ such that the child $v$ has the maximum number of descendants among the children of $u$. In the case of a tie among two children of $u$ pick any one arbitrarily. The edge $(u,v)$ is called a \textit{heavy edge}. Thus, each internal node of $T$ selects exactly one edge as heavy edge.  Further, it is known that each vertex has at most two heavy edges incident on it, one with its parent and second with one of its children.  An edge that is not chosen as a heavy edge by any internal node is called a \textit{light edge}. Consider the forest of paths obtained by removing light edges from $T$. We refer to each path in this forest as  a \textit{heavy path}. The heavy path decomposition of $T$ partitions the nodes of $T$ into the set of heavy paths denoted by $\mathcal{H}$. Also, it partitions the edges of $T$ into heavy and light edges; see Figure \ref{fig:hpd}. 

\begin{figure}[ht]
\centering
\includegraphics[width=1\textwidth]{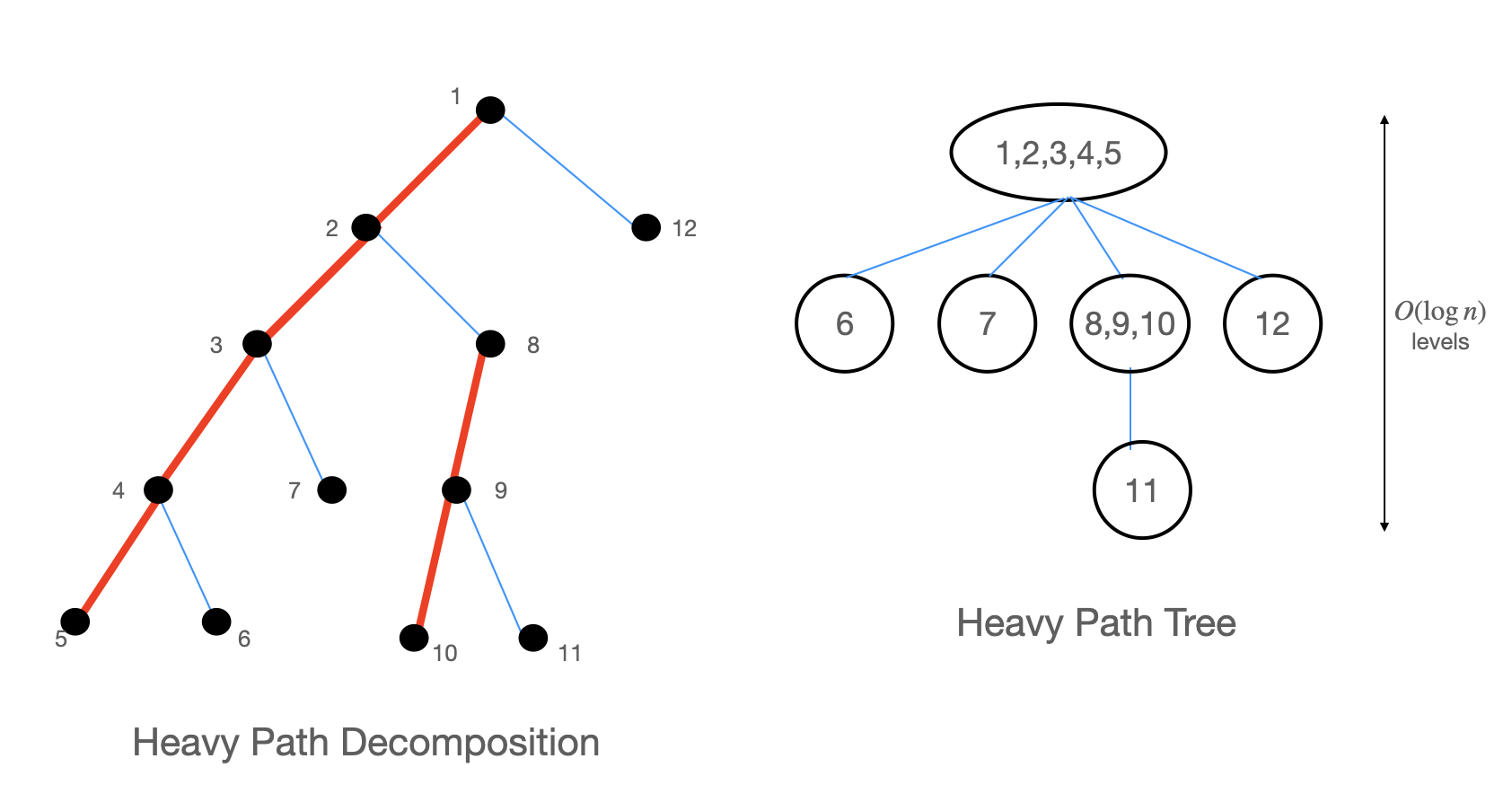}
\caption{Heavy path decomposition of a rooted tree with $n=12$ nodes. At node 1, $\{1,2\}$ is picked as the heavy  edge  as 2 is the child of 1 with maximum number of  descendants. Edge $\{1,12\}$ is a light edge. Path $\{1,2,3,4,5\}$ of the clique tree is compressed as a single node in the heavy path tree $\mathcal{T}$. As a light edge connects a node to a sub-tree that is at least halved in size, the heavy path tree will have at most $\log n$ levels. Thick red lines indicate heavy edges and thin blue lines indicate light edges respectively. }
\label{fig:hpd}
\end{figure}

\begin{remark}
\label{rem:2types}
The heavy paths in ${\cal H}$ are of two types: those that contain at least one heavy edge and those which do not contain a heavy edge.  A heavy path which does not contain a heavy edge is a leaf node whose incident edge in $T$ is a light edge. For instance, in Figure~\ref{fig:hpd} heavy path $\{1,2,3,4,5\}$ is of former type whereas heavy path $\{6\}$ is of later type.
\end{remark}

Using the heavy path decomposition  of $T$ we define a tree which we call the \textit{heavy path tree} denoted by $\mathcal{T}$.  Let $\Phi:\mathcal{H} \rightarrow V(\mathcal{T})$ be a bijection such that for $u \in V(H_1)$ and $v \in V(H_2), e=\{u,v\}$ is a light edge  if and  only if there exists an edge $\{\Phi(H_1), \Phi(H_2)\}$ in $\mathcal{T}$. Subsequently, whenever $H_1$ and $H_2$ satisfy this property we will call $H_1$ and $H_2$ \textit{light edge separable} heavy  paths. In  other words, $H_1,H_2 \in \mathcal{H}$ are light edge separable heavy  paths in $T$ if and  only if $\Phi(H_1)$ and $\Phi(H_2)$ are adjacent in $\mathcal{T}$. Further, we refer to the edge $\{\Phi(H_1),\Phi(H_2)\} \in E({\cal T})$ using the light edge between and $H_1$ and $H_2$, which in this case is $e$.  The heavy path that contains the root of $T$ is the root of $\mathcal{T}$. The level of the root node is 0, and each other node has a level which is its distance from the root. A sub-path  of  a heavy path will be referred to as \textit{heavy sub-path}. The following remark is important for the rest of the  paper. The following lemmata are well-known~\cite{GO}, and we extensively use them here.

\begin{lemma}
\label{lem:hpd_prop}
The number of levels in $\mathcal{T}$ is at most $\lceil \log n \rceil$. Further, a path in $\mathcal{T}$ has at most $2 \lceil \log n \rceil$ edges.
\end{lemma}

\begin{lemma}
\label{lem:vparti}
For $u,v \in V(\mathcal{T}), V(\Phi^{-1}(v)) \cap V(\Phi^{-1}(u)) = \phi$ and $V(T)=\bigcup\limits_{v \in V(\mathcal{T})} V(\Phi^{-1}(v))$. In other words, the nodes of $T$ are partitioned among the nodes of $\mathcal{T}$.
\end{lemma}

\begin{lemma}
\label{lem:pathpartition}
Let $P$ be a path in $T$. $P$ can be partitioned into heavy sub-paths $\Pi=\{\pi_1,\ldots,\pi_k\}, 1 \leq k \leq 2 \lceil \log n \rceil + 1$. 
\end{lemma}
\begin{proof}
Consider a path $P=v_1 e_2 v_2 e_2 \ldots v_{l-1} e_l v_l, 1 \le l \le n$ in $T$ where $v_1,\ldots,v_l \in  V(P)$ and $e_2,\ldots,e_l \in E(P)$.  Let $f_1, \ldots, f_{t}$ denote the $t \geq 0$  light edges in $P$ indexed in the order in which they occur in $P$ from $v_1$ to $v_l$.  Then, we consider $P=\pi_1 f_1 \pi_2 f_2 \ldots f_{t} \pi_{t+1}$ such that for each $1 \leq i \leq t+1$, $\pi_i$ is a maximal heavy sub-path in $P$. $f_i$ is the light edge between the last vertex of $\pi_i$ and the first vertex of $\pi_{i+1}$.  For each $1 \leq i \leq t+1$, let $H_i$ denote the heavy path in $\mathcal{H}$ such that $\pi_i$ is a heavy sub-path of $H_i$.  By our convention on the edge label in ${\cal T}$, the edge $\{\Phi(H_i),\Phi(H_{i+1})\}$ is considered to be $f_i$. Thus, for $P$ in $T$ we have a path $P'=\Phi(H_1) f_1 \Phi(H_2) f_2 \ldots f_{t} \Phi(H_{t+1})$ in $\mathcal{T}$. Let $|V(P')|=t+1$ and from Lemma~\ref{lem:hpd_prop} we know that $t \le 2 \lceil \log n \rceil$. Thus, the number of heavy sub-paths, $k \le 2 \lceil \log n \rceil + 1$. 
\end{proof}

\noindent
The following propositions are straightforward and merely stated explicitly in the context of heavy path trees. 
\begin{proposition}
\label{prop:2node}
Let $P$ be a path in $\mathcal{T}$ and $l \le \log n$ be an integer. $P$ consists of at most two nodes with level $l$.
\end{proposition}


\noindent
Through out this paper we will use the notation $[n]=\{1,2,\ldots,n\}$. 

\subsection{Useful Succinct Data Structures}
Table \ref{table:udstable} summarises the set of data structures we use in this work which we will explain in this section starting with ordinal trees. 

\noindent
{\bf Succinct Data Structure for Ordinal Trees.} Let the children of any $u \in V(T)$ be $\{u_1,\ldots,u_z\}$ for some $z>0$. Tree $T$ is called an \textit{ordinal tree} if for $i < j$, $u_i$ is to the left of  $u_j$~\cite{Nav_Sada}.  By considering ordinal trees as balanced  parenthesis Navarro and Sadakane~\cite{Nav_Sada} has given a $2n+o(n)$ bit succinct data structure. 

\begin{table}[h!]
\centering
\begin{tabular}{ |p{2cm}||p{3.7cm}|p{4.5cm}|p{2.5cm}|p{1.5cm}| }
  \hline
 {\bf Data Structure}& {\bf Query} & {\bf Functionality} & {\bf To store} & {\bf Ref} \\
 \hline
    \multirow{3}{4em}{ Ordinal trees} & $\texttt{lca}(u,v)$ & returns lowest common ancestor of nodes $u$ and $v$ & the clique tree in Sections \ref{sec:succrep} and \ref{sec:seds} & \cite{Nav_Sada} \\ &  & & & \\ 
    & $\texttt{parent}(u)$ & returns parent of node $u$ &  & \cite{Nav_Sada} \\ &  & & & \\ 
    & $\texttt{first\_child}(u)$ & returns first child of node $u$ &  & \cite{Nav_Sada} \\ &  & & & \\ 
    & $\texttt{rmost\_child}(u)$ & returns rightmost leaf of the sub-tree rooted at $u$ &  & \cite{Nav_Sada} \\ &  & & & \\
    & $\texttt{child\_rank}(u)$ & returns the number of siblings to the left of $u$ &  & \cite{Nav_Sada} \\ &  & & & \\
    \hline
    \multirow{3}{4em}{ Bit string} & $\texttt{rank}(B,b,i)$ & returns the number of bit $b$'s up to and including position $i$ on bit vector $B$ from left & for instance, the BP representation of clique tree & \cite{RamanRS07} \\ &  & & & \\
    & $\texttt{select}(B,b,i)$ & returns the position of the $i-$th bit $b$ in the bit vector $B$ from left &  & \cite{RamanRS07} \\ 
    \hline
    \multirow{1}{4em}{ Increasing number sequence} & $\texttt{accessNS}(B,i)$ & returns the $i-$th number in the sequence $B$ & the starting nodes of paths in clique tree in Section \ref{sec:succrep} & Section 2.8 of \cite{Navarro} \\ &  & & & \\
    \hline
    \multirow{1}{5em}{ Wavelet tree} & $\texttt{access}(S, c)$ & returns the $y-$coordinate of the point with $x-$coordinate value $c$ stored in wavelet  tree $S$ & the paths as points in a two dimensional grid in Section \ref{sec:succrep}  & \cite{Makinen2007} \\ &  & & & \\ 
    & $\texttt{select}(S, [i,i'],[j,j'])$ & returns the points in the input range $[i,i'] \times [j,j']$ &  & \cite{Makinen2007} \\ &  & & & \\ 
    & $\texttt{count}(S, [i,i'],[j,j'])$ & returns the number of points in the input range $[i,i'] \times [j,j']$ &  & \cite{Makinen2007} \\ 
    \hline
\end{tabular}
\caption{Summary of the data structures used.  Note that $b \in \{0,1\}$.}
\label{table:udstable}
\end{table}

\begin{lemma}
\label{lem:2ntree}
For any ordinal tree $T$ with $n$ nodes, there exists a $2n+o(n)$ bit Balanced Parentheses (BP) based data structure that supports the following four functions among others in constant time :
\begin{enumerate}
    \itemsep0em 
    \item ${\normalfont\texttt{lca}}(i,j)$, returns the lowest common ancestor of two nodes $i,j$ in $T$, 
    \item ${\normalfont\texttt{parent}}(i)$, returns the parent of node $i$ in $T$, and 
    \item ${\normalfont\texttt{first\_child}}(i)$, returns the first child of node $i$ in $T$.
    \item ${\normalfont\texttt{rmost\_leaf}}(i)$, returns the rightmost leaf of sub-tree rooted at node $i$ in $T$.
    \item ${\normalfont\texttt{child\_rank}}(i)$, returns the number of siblings to the  left of node $i$ in $T$.
\end{enumerate}
\end{lemma}

\noindent
{\bf Rank-Select Data Structure.} Bit-vectors are extensively used in the succinct representation given in Section~\ref{sec:succrep}. The following data structure due to Golynski et al. \cite{GMS2006} and the functions supported by  it are useful.

\begin{lemma}
\label{lem:bstr}
Let $B$ be an $n-$bit vector and  $b \in \{0,1\}$.  There exists an $n+o(n)$ bit data structure that supports the following functions in constant time:
\begin{enumerate}
    \itemsep0em 
    \item ${\normalfont\texttt{rank}}(B,b,i)$:  Returns the number of $b$'s up to and including position $i$ in the bit vector $B$ from the left. 
    \item ${\normalfont\texttt{select}}(B,b,i)$: Returns the position of the $i$-th $b$ in the bit vector $B$ from left. For $i \notin [n]$ it returns 0.
\end{enumerate}
\end{lemma}

\noindent
{\bf Non-decreasing Integer Sequence Data Structure.} Given a set of positive integers in the non-decreasing order we can store them efficiently using the differential encoding scheme for increasing numbers; see Section 2.8 of \cite{Navarro}. Let $S$ be the data structure that supports differential encoding for increasing numbers then the function $\texttt{accessNS}(S,i)$ returns the $i-$th number in the sequence. 

\begin{lemma}
\label{lem:ssn}
Let $S$ be a sequence of $n$ non-decreasing positive integers $a_1,\ldots,a_n, 1 \leq a_i \leq n$. There exists a $2n+o(n)$ bit data structure that supports ${\normalfont\texttt{accessNS}}(S,i)$ in constant time.
\end{lemma}
\begin{proof}
We will prove the lemma by giving a construction of such a data structure. $a_1$ will be represented by a sequence of $a_1$ 1's followed by a 0. Subsequently $a_i$'s are represented by storing $a_i-a_{i-1}$ many 1's followed by a 0. It will take $2n$ bits since there are $n$ 0's and $n$ 1's. Let this bit string be  stored using the data structure of Lemma~\ref{lem:bstr} and be denoted as $B$. $B$ takes $2n+o(n)$ bits. $\texttt{accessNS}(S,i)$ can be implemented using $\texttt{rank}(B,1,\texttt{select}(B,0,i))$ on the bit string obtained.
\end{proof}


\noindent
{\bf Wavelet Trees.} Central to the design of the succinct data structure of Section \ref{sec:succrep}  is the two-dimensional orthogonal range search data structure used to store points in the two-dimensional plane. Specifically, we use the $n \log n + o(n \log n)$ bit succinct \textit{wavelet trees} due to Makinen and Navarro~\cite{Makinen2007} that requires the $n$ points to have distinct integer-valued $x-$ and $y-$coordinates in the range $[n] \times [n]$. The wavelet  tree has the following properties.
\begin{enumerate}
    \itemsep0em 
    \item The wavelet tree is a balanced binary search tree. Each node of the tree is associated with an interval range of the $x$-coordinate.
    \item The range at the root of the wavelet tree has the interval $[1,n]$ and the interval at each leaf is of the form $[a,a], 1 \le a \le n$.
    \item The range  $[a,a'], 1 \le a < a' \le n,$ at an internal node is partitioned among the ranges $[a,a'']$ and $[a''+1,a']$ at its children, that is, $[a,a'] = [a,a''] \bigcup [a''+1,a']$. 
\end{enumerate}
    
\noindent
We use the following result regarding wavelet trees from~\cite{Makinen2007}.

\begin{lemma}
\label{lem:rsds}
Given a set of $n$  points $\{(a_1,b_1),\ldots,(a_n,b_n)\}$ where $a_i,b_i \in [n], 1 \le i \le n$ such that $a_i \ne a_j$ and $b_i \ne b_j$ for $i \ne j$, there exists an $n \log n (1 + o(1))$ bit orthogonal range search data structure $S$ that supports the following functions:
\begin{enumerate}
    \itemsep0em 
    \item ${\normalfont\texttt{search}}(S, [i,i'],[j,j'])$: Returns the points in the input range $[i,i']\times[j,j'], 1 \le i \le i' \le n, 1 \le j \le j' \le n,$ in the increasing order of the $x-$coordinate taking $O(\log n)$ time per point.
    \item ${\normalfont\texttt{count}}(S, [i,i'],[j,j'])$: Returns the number of points in the input range $[i,i']\times[j,j'], 1 \le i \le i' \le n, 1 \le j \le j' \le n,$ in $O(\log n)$ time.
    \item ${\normalfont\texttt{access}}(S, i):$ Returns the $y-$coordinate of the point stored in  $S$  with $x-$coordinate $i$ in $O(\log  n)$ time.
\end{enumerate}
\end{lemma}

\section{Clique Tree Pre-processing = HPD + Pre-order Traversal}
\label{sec:proprocessingT}
The pre-processing of the clique tree allows the paths in $\mathcal{P}$ to be stored efficiently so that adjacency and neighbourhood queries can be supported. This involves the following two steps:
\begin{enumerate}
    \itemsep0em 
    \item heavy path decomposition  of $T$, and
    \item transformation of $T$ into  an ordinal tree which is labeled based on the pre-order traversal
\end{enumerate}

\noindent
A clique tree is organized as an ordinal tree as explained below.

\noindent
{\bf HPD + Pre-order traversal of $T$.} Fix a root node for $T$ and perform heavy path decomposition on it. For $v \in V(T)$ order its children $(w_1,\ldots,w_c)$ such that $\{v,w_1\}$ is a heavy edge. Let the children adjacent to $v$ by light edges $(w_2,\ldots,w_c)$, be ordered arbitrarily. This ordering of children of a node of the clique tree makes it an ordinal tree. Label the nodes of this ordinal tree based on the  pre-order traversal; see Section 12.1 of~\cite{CLRS} for more details of pre-order traversal of trees. Labels assigned to nodes in this manner are called the \textit{pre-order label of the nodes}. Through out the rest of our paper this ordinal rooted clique tree labeled with pre-order will be referred to as the clique tree.

\noindent
{\bf Representing paths as tuples.} Path $P \in \mathcal{P}$ is represented as $P=(l,r), l,r \in  V(T)$, where $l$ and $r$ are the end points of $P$ such that $l \le r$. We say that $l$ and $r$ are the starting and ending nodes of the path, respectively. Let $T_a$ and $T_b$ be the sub-trees of  $T$ rooted at $a$ and  $b$, respectively. For $e=\texttt{lca}(l,r)$, the following propositions regarding the clique tree $T$ follow from pre-order traversal. 

\begin{proposition}
\label{prop:pathnsubtree}
Let $b \in V(T_a)$ for $a \in V(T)$. The following hold.
\begin{enumerate}
    \itemsep0em 
    \item The sub-tree rooted at $b$ is contained in $T_a$.
    \item Let $V(P) \nsubseteq V(T_a)$. $|V(P) \cap V(T_a)| \ge 1$ if and  only  if $a \in V(P)$.    
    \item $l \in V(T_a)$ and $e < a$ if and only  if  $r \notin V(T_a)$ and $a \in V(P)$.
    \item If $l \in V(T_a)$ and $e = a$ then  $r \in V(T_a)$.
    \item If there are $z$ descendants of $a$ in $T_a$ then the nodes of $V(T_a)$ is
the set $\{a, a+1, \ldots, a + z\}$.
\end{enumerate}
\end{proposition}

\begin{proposition}
\label{prop:consecu}
Let $\pi$ be a heavy path of length $k \ge 1$ in $T$ with $a$ and $b$ as its end points such that $a \le b$. Then $b=a+k$. 
\end{proposition}


\noindent
We will subsequently use the notation $[a,b]$  to refer to the ordered set of nodes $(a,a+1,\ldots,b-1,b)$.  $a$ and $b$ are referred to as the starting and ending points of $[a,b]$.  Figure~\ref{fig:pathgraph} shows  the  pre-processed clique tree for the example path graph $G$, also shown in the figure. The heavy path starting at 1 and ending at 4 have contiguous numbering and is denoted as $[1,4]$. We emphasize that a sub-path of a heavy path is also represented using the same notation.   A heavy path or a heavy sub-path $[a,a]$ contains only the vertex $a$.    

\begin{lemma}
\label{lem:hpd_prop1}
For any heavy path $H=[a',b']$ of $T$ a path $\pi=[a,b]$ such that $a' \le a \le b \le b'$ is a heavy sub-path of $H$. For the heavy sub-path $\pi$ of $H$ let $c$ and $d$ be the rightmost leaves in the sub-tree rooted at $b$ and $a$, respectively. The following are true about $\pi$. 
\begin{enumerate}
    \itemsep0em 
    \item If  $a \le u \le b$ then $u \in V(\pi)$.
    \item $a \le b \le c \le d$.
    \item Let $Q=(s,t)$ be a path in $T$ . If $s > d$ then $\pi$ and $Q$ are vertex disjoint paths in $T$.
\end{enumerate}
\end{lemma}
\begin{proof} The proofs are as  follows:
\begin{enumerate} 
    \itemsep0em 
    \item This  is  true due to Proposition~\ref{prop:consecu}.
    \item $a \le b$ by definition of $\pi$. $b \le c \le d$ since the labels are based on pre-order traversal.
    \item From Proposition~\ref{prop:pathnsubtree} we know that, if there are $z$ descendants of $a$ in $T_a$, then the nodes of $V(T_a)$ is
    the set $\{a, a+1, \ldots, a + z\}$.  Since $d$ is the label of the rightmost descendant of $a$, it follows that $d = a+z$. Since $a+z = d < s$, we know that $Q$ starts at a node that is visited after the nodes of $T_a$. Since $t>s$ the pre-order labels of nodes in $Q$ is not in $V(T_a)$. Thus, $Q$ and $T_a$ are vertex disjoint, and thus $Q$ and $\pi$ are vertex disjoint.
\end{enumerate}
\end{proof}

  
\begin{lemma}
\label{lem:lca_range}
\begin{enumerate} Let $a \in V(T)$ and $a_l$ be a child of $a$ such that $l \in V(T_{a_l})$ . 
    \item If ${\normalfont \texttt{lca}(l,r)} < a$ then $r \in [{\normalfont \texttt{rmost\_leaf}(a)+1,n}]$.
    \item If ${\normalfont \texttt{lca}(l,r)} = a$ then $r \in [{\normalfont \texttt{rmost\_leaf}(a_l)+1,\texttt{rmost\_leaf}(a)}]$.
\end{enumerate}
\end{lemma}
\begin{proof} The proof is as follows:
    \begin{enumerate}
        \itemsep0em 
        \item Since $\texttt{lca}(l,r) < a$ and $r \ge l$,  $r$ is a node that is visited after the nodes in $T_a$ are visited, that is, $r \in [\texttt{rmost\_leaf}(a)+1,n]$. 
        \item If $\texttt{lca}(l,r) = a$ then there exists a child of $a$, say $a_r$ such that $r \in T_{a_r}$. $a_l < a_r$ since $l < r$. Thus, $r \in  [\texttt{rmost\_leaf}(a_l)+1,\texttt{rmost\_leaf}(a)]$.
    \end{enumerate}    
\end{proof}

\subsection{Organizing the heavy paths and light edges of $T$} 
\label{sec:totalorders}
Let $\mathcal{H}$ be the set of heavy paths of $T$; recall from Section~\ref{sec:hpd}. Let $H, H' \in \mathcal{H}$ such  that $H = [a,b]$ and $H' = [a',b']$. We define a total order $(\mathcal{H},\prec_{\mathcal{H}})$ as follows.  $H \prec_{\mathcal{H}} H'$ if $a < a'$. In other words, $H \prec_{\mathcal{H}} H'$ if $H$ is visited before $H'$ in the pre-order traversal  of  $T$.

\noindent
{\bf Total order on the heavy sub-paths of paths in $T$.}  $\prec_{\mathcal{H}}$, extends to the set of heavy sub-paths of  a path. Let $\Pi=\{\pi_1,\ldots,\pi_k\}, 1 \le k \le 2 \lceil \log n \rceil +  1$, be the set of heavy sub-paths of path $P$; see Lemma~\ref{lem:pathpartition} for details regarding heavy sub-paths of a path. For any two $\pi,\pi' \in \Pi$, let $H, H' \in {\cal H}$ be such that $\pi$ and $\pi'$ are heavy sub-paths of $H$ and $H'$, respectively.
$\pi \prec \pi'$ if $H \prec_{\mathcal{H}} H'$. In  other words, we order the heavy sub-paths according to the order of the heavy paths that contain  it. 


\noindent
{\em Convention:} In the rest of this section, $P$ denotes the path $(l,r)$ in $T$, and $\Pi$ is the decomposition of $P$ into heavy sub-paths. In other words, for the path $P$, $\Pi=(\pi_1,\ldots,\pi_k), 1 \le k \le 2 \lceil \log n \rceil + 1, \pi_i \prec \pi_j$ if $1 \le i < j \le k$. Also, for every $\pi_i \in \Pi$,  $\pi_i=(a_i,b_i)$. In Section~\ref{sec:seds}, we assume that the heavy paths of $T$ are numbered such that $H_i \prec_{\mathcal{H}} H_j$ if and only if $1 \le i  < j \le n$.

\subsection{Characterising path intersections in $T$}
\label{sec:charpathintersection}
In this section, first we will show that a heavy sub-path $\pi \in \Pi$ partitions the nodes of $T$ into four ranges of pre-order labels. Paths intersecting $\pi$ are characterised based on these ranges. Adjacency and neighbourhood queries for $\pi$ are implemented using orthogonal range search queries that use these ranges.

\noindent
{\bf  Successor of a heavy sub-path.} For $\pi_i, \pi_j \in \Pi$, if there exists nodes $u_1 \in V(\pi_i)$ and $u_2 \in V(\pi_j)$ such that $\{u_1,u_2\}$  is a light edge in $T$ then we say  that $\pi_i$ and $\pi_j$ are \textit{light edge separable} heavy  sub-paths; an extension of  the notion of light edge separable heavy paths from Section~\ref{sec:hpd}. Let $1 \le i < j \le k, \pi_i,\pi_j \in \Pi$. $\pi_j$ is called the \textit{successor} of $\pi_i$   if  $\pi_i \prec \pi_j$ and they are light edge separable. We define a mapping $\texttt{succ}:\Pi \times \{1,2\} \rightarrow \Pi \cup \{\texttt{NULL}\}$ from a heavy sub-path of $P$ to its successors defined as follows.
\begin{enumerate}
    \itemsep0em 
    \item Successors for $\pi_1$: We have the following sub-cases depending on  the number of heavy sub-paths of $P$. 
    \begin{enumerate}
        \itemsep0em 
        \item $k = 1:$  $\pi_1$ has no successors, that is, $\texttt{succ}(\pi_1,1)=\texttt{succ}(\pi_1,2)=\texttt{NULL}$.     
        \item $k > 1:$ There are two cases:
        \begin{enumerate}   
            \itemsep0em 
            \item $\pi_1$ has both its successors that is $\texttt{succ}(\pi_1,1) = \pi_2 \ne \texttt{NULL}$ and $\texttt{succ}(\pi_1,2) = \pi_z \ne \texttt{NULL}$ for $z \in [3,k]$.  This happens when $P$ can be divided into two sub-paths $P^1=(e,l)$ and $P^2=(e,r)$ such that $e \ne l \ne r$.
            \item $\pi_1$ has only one successor. Let $\pi_2=(a_2,b_2)$ and $a_1 \ne b_1$. There are two sub-cases.
            \begin{enumerate}
                \itemsep0em 
                \item $\texttt{succ}(\pi_1,1)=\texttt{NULL}$ and $\texttt{succ}(\pi_1,2)=\pi_2$. This happens when $\texttt{parent}(a_2)=a_1$.
                \item $\texttt{succ}(\pi_1,1)=\pi_2$ and $\texttt{succ}(\pi_1,2)=\texttt{NULL}$. This happens when $\texttt{parent}(a_2)=b_1$.
            \end{enumerate}
            When $a_1=b_1$, we define $\texttt{succ}(\pi_1,1)=\pi_2$ and $\texttt{succ}(\pi_1,2)=\texttt{NULL}$.
        \end{enumerate}
    \end{enumerate}
    \item Successors for $\pi_i, i \ne 1$:  $\texttt{succ}(\pi_i,1)=\pi_{i+1}$ if $\pi_{i+1} \in \Pi$ and  $\pi_i,\pi_{i+1}$ are light edge separable  else $\texttt{succ}(\pi_i,1)=\texttt{NULL}$. For all $i \ne 1, \texttt{succ}(\pi_i,2)=\texttt{NULL}$.
\end{enumerate}

\noindent
Note that if $\texttt{succ}(\pi_1,2) = \pi_z \ne \texttt{NULL}$ for $z \in [3,k]$ then $\texttt{succ}(\pi_{z-1},1)=\texttt{succ}(\pi_{z-1},2)=\texttt{NULL}$. Also, $\texttt{succ}(\pi_k,1)=\texttt{succ}(\pi_k,2)=\texttt{NULL}$.

\noindent
{\bf Interval ranges associated with $\pi_i \in \Pi, 1 \le i \le k$.}  For $u \in V(T)$, $\texttt{rmost\_leaf}(u)+1$ is the node that is visited immediately after traversing the nodes in sub-tree rooted at $u$ in the  pre-order traversal of $T$.  We associate four ranges of nodes of $T$ with $\pi_i$.  The four ranges associated with $\pi_i$ denoted by $R_j(i), 1 \le j \le 4,$ are as follows.
\begin{enumerate}[i.]
    \itemsep0em 
    \item $R_1(i)$: Range of nodes visited before $a_i$ in the pre-order traversal of $T$. If $i=1$ and $a_1 > 1$ then $R_1(1)=[1,a_1-1]$ else $R_1(1)=\phi$.

    \item $R_2(i)$: Heavy sub-path $\pi_i$, $[a_i,b_i]$.
    \item $R_3(i)$: There are two cases depending on existence of $\texttt{succ}(\pi_i,1)$.
    \begin{enumerate}
        \itemsep0em 
        \item If $\texttt{succ}(\pi_i,1) \ne \texttt{NULL}$ then $R_3(i)=R_3^1(i) \cup R_3^2(i)$ where
        \begin{enumerate}
            \itemsep0em 
            \item $R_3^1(i)$: Range of nodes visited after $b_i$ and before the nodes of $T_{a_{i+1}}$, $R_3^1(i)=[b_i+1,a_{i+1}-1]$.
            \item $R_3^2(i)$: Range of nodes visited after visiting the nodes of  $T_{a_{i+1}}$ and before the right-most leaf of $T_{b_i}$. If $\texttt{rmost\_leaf}(a_{i+1}) \ne \texttt{rmost\_leaf}(b_i)$ then $R_3^2(i)=[\texttt{rmost\_leaf}(a_{i+1})+1,\texttt{rmost\_leaf}(b_i)]$ else $R_3^2(i)=\phi$.
        \end{enumerate}
        \item If $\texttt{succ}(\pi_i,1) = \texttt{NULL}$ and $b_i \ne \texttt{rmost\_leaf}(b_i)$ then $R_3(i)=[b_i+1,\texttt{rmost\_leaf}(b_i)]$ else $R_3(i)=\phi$. Note that $b_i \ne \texttt{rmost\_leaf}(b_i)$ means that $b_i$ is not a  leaf.
    \end{enumerate} 
    \item $R_4(i)$: There are two cases depending on  $\texttt{succ}(\pi_i,2)$.
    \begin{enumerate}
        \itemsep0em 
        \item If $\texttt{succ}(\pi_i,2) \ne \texttt{NULL}$ then by definition $i=1$ and $R_4(1)=R_4^1(1) \cup R_4^2(1)$ where
        \begin{enumerate}
            \itemsep0em 
            \item $R_4^1(1)$: Range of nodes visited after nodes in $T_{b_1}$ and before nodes of $T_{a_z}$. If $\texttt{rmost\_leaf}(b_1)+1 \ne a_z$ then $R_4^1(1)=[\texttt{rmost\_leaf}(b_1)+1,a_z-1]$ else $R_4^1(1)=\phi$. Note that $\texttt{rmost\_leaf}(b_1)+1 \ne a_z$ means $b_1$ is not a leaf.
            \item $R_4^2(1)$: Range of nodes visited after nodes in $T_{a_z}$ and before the right-most node of $T_{a_1}$. If $\texttt{rmost\_leaf}(a_z) \ne \texttt{rmost\_leaf}(a_1)$ then $R_4^2(1)=[\texttt{rmost\_leaf}(a_z)+1,\texttt{rmost\_leaf}(a_1)]$  else $R_4^2(1)=\phi$. 
        \end{enumerate}
        \item If $\texttt{succ}(\pi_i,2) = \texttt{NULL}$ and $\texttt{rmost\_leaf}(a_i) \ne b_i$ then $R_4(i)=[\texttt{rmost\_leaf}(b_i)+1, \texttt{rmost\_leaf}(a_i)]$  else $R_4(i)=\phi$.
    \end{enumerate}
\end{enumerate}

\noindent
{\em Remark.} For each $1 \le i \le k$, $1 \le j \le 4, R_j(i) \subseteq [1,n]$ are intervals, and
 $R_1(i) \cup R_2(i) \cup R_3(i) \cup R_4(i) = V(T_{a_i}) \cup [1,a_i-1]$. Further, the range of nodes greater than $\texttt{rmost\_leaf}(a_i)$ is not relevant, as it follows from Lemma~\ref{lem:hpd_prop1} that the starting point of a path intersecting with $P$ should be in one of these ranges. See Figure~\ref{fig:heavysubpath_regions} for a pictorial representation of the ranges and Table~\ref{table:rangetable1} and~\ref{table:rangetable2} summarise the ranges.

%

\begin{figure}[ht]
\centering
\includegraphics[width=1\textwidth]{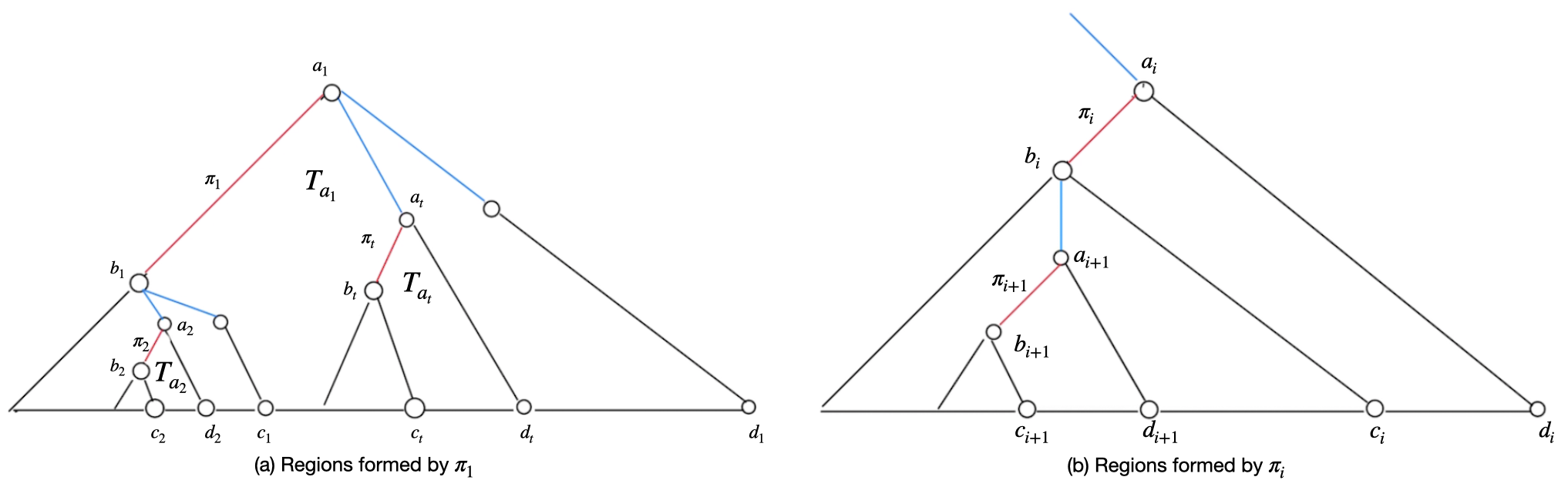}
\caption{The regions generated  by (a) $\pi_1$, and (b) $\pi_i, 2 \le i \le n$. For $j \in \{1,2,t,i,i+1\}$, $\texttt{rmost\_leaf}(a_j)=d_j$ and $\texttt{rmost\_leaf}(b_j)=c_j$.}
\label{fig:heavysubpath_regions}
\end{figure}

\begin{table}[!h]
\centering
\begin{tabular}{|p{1cm}||p{1cm}|p{4cm}|p{8cm}| }\hline
 {}& {} & {\bf Condition} & {\bf Range } \\\cline{2-3}\cline{3-4}
\multirow{7}{*}{$i=1$} & & & \\
    & $R_1$ & $a > 1$ & $R_1(1)=[1,a_1-1]$\\
    & & & \\\cline{3-4}
    & & & \\
    & & $a=1$ & $R_1(1)=\texttt{NULL}$\\
    & & & \\\cline{2-3}\cline{3-4}
    & & & \\
    & $R_2$ &  & $R_2(1)=[a_1,b_1]$\\
    & & & \\\cline{2-3}\cline{3-4}
    & & & \\
    & $R_3$ & $\texttt{succ}(\pi_1,1) \ne \texttt{NULL}$ & $R_3(1)=R_3^1(1) \cup R_3^2(1)$\begin{enumerate}
                    \item $R_3^1(1)$:  $R_3^1(1)=[b_1+1,a_2-1]$.
                    \item $R_3^2(1)$: If $\texttt{rmost\_leaf}(a_2) \ne \texttt{rmost\_leaf}(b_1)$ then $R_3^2(1)=[\texttt{rmost\_leaf}(a_2)+1,\texttt{rmost\_leaf}(b_1)]$ else $R_3^2(1)=\texttt{NULL}$.
                \end{enumerate} \\
    & & & \\\cline{3-4}
    & & & \\
    &       & $\texttt{succ}(\pi_1,1) = \texttt{NULL}$ & If $\texttt{rmost\_leaf}(a_2) \ne \texttt{rmost\_leaf}(b_1)$ then $R_3^2(1)=[\texttt{rmost\_leaf}(a_2)+1,\texttt{rmost\_leaf}(b_1)]$ else $R_3^2(1)=\texttt{NULL}$ \\
    & & & \\\cline{2-3}\cline{3-4}
    & & & \\
    
    & $R_4$  & $\texttt{succ}(\pi_1,2) \ne \texttt{NULL}$ & $R_4(1)=R_4^1(1) \cup R_4^2(1)$ \begin{enumerate}
                \item $R_4^1(1)$: If $\texttt{rmost\_leaf}(b_1)+1 \ne a_t$ then $R_4^1(1)=[\texttt{rmost\_leaf}(b_1)+1,a_t-1]$ else $R_4^1(1)=\texttt{NULL}$.
                \item $R_4^2(1)$: If $\texttt{rmost\_leaf}(a_t) \ne \texttt{rmost\_leaf}(a_1)$ then $R_4^2(1)=[\texttt{rmost\_leaf}(a_t)+1,\texttt{rmost\_leaf}(a_1)]$  else $R_4^2(1)=\texttt{NULL}$.
            \end{enumerate} \\\cline{3-4}
    & & & \\        
    &       & $\texttt{succ}(\pi_1,2) = \texttt{NULL}$ & If $\texttt{rmost\_leaf}(a_1) \ne b_1$ then $R_4(1)=[\texttt{rmost\_leaf}(b_1)+1, \texttt{rmost\_leaf}(a_1)]$ else $R_4(1)=\texttt{NULL}$\\
    & & & \\\hline
\end{tabular}
\caption{Ranges for heavy sub-path $\pi_1 \in  \Pi$.}
\label{table:rangetable1}
\end{table}

\begin{table}[!h]
\centering
\begin{tabular}{|p{1cm}||p{1cm}|p{4cm}|p{8cm}| }\hline
 {}& {} & {\bf Condition} & {\bf Range } \\\cline{2-3}\cline{3-4}
\multirow{7}{*}{$i \ne 1$} & & & \\
    & $R_1$ &  & $R_1(i)=\texttt{NULL}$\\
    & & & \\\cline{2-3}\cline{3-4}
    & & & \\
    & $R_2$ &  & $R_2(i)=[a_i,b_i]$\\
    & & & \\\cline{2-3}\cline{3-4}
    & & & \\
    & $R_3$ & $\texttt{succ}(\pi_i,1) \ne \texttt{NULL}$ & $R_3(i)=R_3^1(i) \cup R_3^2(i)$\begin{enumerate}
                    \item $R_3^1(i)$:  $R_3^1(i)=[b_i+1,a_{i+1}-1]$.
                    \item $R_3^2(i)$: If $\texttt{rmost\_leaf}(a_{i+1}) \ne \texttt{rmost\_leaf}(b_i)$ then $R_3^2(i)=[\texttt{rmost\_leaf}(a_{i+1})+1,\texttt{rmost\_leaf}(b_i)]$ else $R_3^2(i)=\texttt{NULL}$.
                \end{enumerate} \\\cline{3-4}
    & & & \\            
    &       & $\texttt{succ}(\pi_i,1) = \texttt{NULL}$ & If $\texttt{rmost\_leaf}(b_i) \ne \texttt{rmost\_leaf}(b_i)$ then $R_3(i)=[b_i+1,\texttt{rmost\_leaf}(b_i)]$ else $R_3(i)=\texttt{NULL}$ \\
    & & & \\\cline{2-3}\cline{3-4}
    & & & \\
    & $R_4$  & $\texttt{rmost\_leaf}(a_i) \ne b_i$ &   $R_4(i)=[\texttt{rmost\_leaf}(b_i)+1, \texttt{rmost\_leaf}(a_i)]$ \\
    & & & \\\cline{3-4}
    & & & \\
    & & $\texttt{rmost\_leaf}(a_i) = b_i$ & $R_4(i)=\texttt{NULL}$\\
    & & & \\
    \hline
\end{tabular}
\caption{Ranges for heavy sub-path $\pi_i \in  \Pi$.  Note that for $i \ne 1, \texttt{succ}(\pi_i,2)=\texttt{NULL}$.}
\label{table:rangetable2}
\end{table}

\noindent
{\bf  Characterising intersection of a path $Q$ with a heavy sub-path.} Let $Q=(s,t)$ be a  path in $\mathcal{P}$. Let the sequence of nodes in $Q$ be $(s=u_1,u_2,\ldots,u_z=t), 1 \le z \le n$. Let $y$ be the first node in the sequence such that $y \in V(P)$. 
From Lemma \ref{lem:vparti}, it follows that $\Pi$ partitions the vertices of $P$, and thus $y$ belongs to a unique $\pi \in \Pi$. \\
{\em Convention:} For ease of presentation, in the rest of this section, $Q=(s,t)$ is a path in $\mathcal{P}$, and $y$ denotes the first vertex in $(s=u_1,u_2,\ldots,u_z=t), 1 \le z \le n$ which is in $P$.\\
For  path $P$ and  path $Q \in \mathcal{P}$ we define the many-to-one function $\alpha:\mathcal{P} \times \mathcal{P} \rightarrow \Pi \cup \{\texttt{NULL}\}$ as follows.

$$
\alpha(P,Q)=\begin{cases}
            \texttt{NULL}, & \text{if }V(P) \cap V(Q) = \phi \\
            \pi \in  \Pi,  & \text{ if } y \in V(\pi)
		 \end{cases}
$$

\noindent
As a consequence of the definition of $\alpha(P,Q)$ we have the following lemma. 

\begin{lemma}
\label{lem:FFconditions}
For $1 \le i \le k$, $\alpha(P,Q)=\pi_i$ if and only if  exactly one of the following is true.
\begin{enumerate}
    \itemsep0em 
    \item $i=1$ and $s \in R_1(1)$ and $t \in V(T_{a_1})$.
    \item $s \in R_2(i)$
    \item $s \in R_3(i)$ and ${\normalfont \texttt{lca}}(s,t) \le b_i$  
    \item $s \in R_4(i)$ and ${\normalfont \texttt{lca}}(s,t) < b_i$ 
\end{enumerate} 
\end{lemma}
\begin{proof}
Let $\alpha(P,Q)=\pi_i$ then the position of the starting node of $Q$ has three possibilities relative to $\pi_i=(a_i,b_i)$. They are as follows:
\begin{enumerate}
    \itemsep0em 
    \item $s < a_i$: This can happen only when $i=1$. By  Proposition~\ref{prop:pathnsubtree}, $a_1 \in V(Q)$ and $y=a_1$. In  this case, $s \in R_1(1)$ and $t \in V(T_{a_1})$. For $i \ne 1$, any  path $Q$ with $l \in R_1(i)$ and $t \in T_{a_i}$ will have to pass through $b_{i-1}$. This implies $\alpha(P,Q) \ne \pi_i$ and thus a contradiction. 
    \item $a_i \le s \le b_i$: By Lemma~\ref{lem:hpd_prop1}, $s = y \in [a_i,b_i]$ that is $s \in R_2(i)$. 
    \item $s > b_i$: In this case, $s \ne y$ and  $y \in [a_i,b_i]$. Depending on the regions of $\pi_i$ as described above we have  two  possibilities as shown below:
    \begin{enumerate}
        \itemsep0em 
        \item $s \in R_3(i)$ and $\texttt{lca}(s,t) \le b_i$. By Proposition~\ref{prop:pathnsubtree}, $b_i \in V(Q)$  and $y = b_i$. 
        \item $s \in R_4(i)$ and $\texttt{lca}(s,t) < b_i$. In this case, $y \in [a_i,b_i-1]$.
    \end{enumerate}
\end{enumerate}
Exactly one of the conditions is satisfied as the ranges $R_j(i)$ are non-overlapping and paths  start in any one  of the ranges. On the other hand, if there exists an $i$ such  that any one of the four conditions as given below is true, then we show that $\alpha(P,Q)=\pi_i$. 
\begin{enumerate}
    \itemsep0em 
    \item $s \in R_1(1)$ and $t \in V(T_{a_1})$: In this case, $y=a_1$ and since $y \in V(\pi_1)$, $\alpha(P,Q)=\pi_1$.
    \item $s \in R_2(i)$: In this case, $y=s$ and since $y \in V(\pi_i)$, $\alpha(P,Q)=\pi_i$.
    \item $s \in R_3(i)$ and $\texttt{lca}(s,t) \le b_i$: Since $l \in R_3(i)$ and $\texttt{lca}(l,r) \le b_i$, $y=b_i$. Since $y \in V(\pi_i)$, $\alpha(P,Q)=\pi_i$.
    \item $s \in R_4(i)$ and $\texttt{lca}(s,t) < b_i$: Since $s \in R_4(i)$ and $\texttt{lca}(s,t) < b_i$, $y \in [a_i,b_i-1]$. Since $y \in V(\pi_i)$, $\alpha(P,Q)=\pi_i$. 
\end{enumerate}
\end{proof}

\noindent
{\bf Function $\texttt{check}$$\alpha$}. This is a useful function that returns true if $\alpha(P,Q)=\pi_i, \pi_i \in \Pi,$ based on conditions of Lemma~\ref{lem:FFconditions}. 

\begin{lemma}
\label{lem:checkalpha}
For path $P \in \mathcal{P}$, given as input the index $i$ of a heavy sub-path $\pi_i \in \Pi$, successors of $\pi_1 \in \Pi$ and another path $Q \in \mathcal{P}$, there exists a function \\ ${\normalfont \texttt{check}\alpha(i,Q,\Pi,\texttt{succ}(1,1),\texttt{succ}(1,2))}$ that checks $\alpha(P,Q)=\pi_i$ in constant time.
\end{lemma}
\begin{proof}
The check can be done in the following manner.
\begin{enumerate}
    \itemsep0em 
    \item Compute the interval ranges $R_j, 1 \le j \le 4,$ of $\pi_i$ using its end points and its successor stored in  $\Pi$. If $i=1$ then we can get the successors $\texttt{succ}(1,1)$ and $\texttt{succ}(1,2)$ from the input. If $i \ne 1$ then it can be obtained from $\Pi$  as follows. For $i \ne 1$, it is $\pi_{i+1}$ unless $\pi_{i+1}=\texttt{succ}(1,2)$ or $i=k$ where $k$ is the number of heavy sub-paths in $\Pi$. Since there are only four ranges and from  Lemma~\ref{lem:2ntree}, $\texttt{rmost\_leaf}$ takes constant time, the ranges can be computed in constant time.
    \item Check if $\alpha(P,Q)=\pi_i$ based on Lemma~\ref{lem:FFconditions}. It takes constant time as the  four checks are based on comparisons and from Lemma~\ref{lem:2ntree}, $\texttt{lca}(s,t)$ can be computed in constant time. If $\alpha(P,Q)=\pi_i$ then $\texttt{check}$$\alpha$ returns true.
\end{enumerate}
Thus, $\texttt{check}$$\alpha$ checks $\alpha(P,Q)=\pi_i$ in constant time. 
\end{proof}

\noindent
For $\pi \in \Pi$, let $\beta(\pi)=\{Q \mid Q \in \mathcal{P} \text{ and } \alpha(P,Q)=\pi \}$.  We have the following lemma.

\begin{lemma}
\label{lem:disjointF}
For all distinct $\pi, \pi' \in \Pi, \beta(\pi) \cap \beta(\pi') = \phi$.
\end{lemma}
\begin{proof}
For each $\pi \in \Pi$, $\beta(\pi)$ is the pre-image of $\pi$ under the function $\alpha$.  Since $\alpha$ is a function, it follows that if $\pi \neq \pi'$, $\beta(\pi) \cap \beta(\pi') = \phi$.
\end{proof}

\noindent
Let the \textit{neighbourhood of a path} be the set of all paths that have non-empty intersection with it. We have the following theorem regarding neighbourhood.
\begin{lemma}
\label{lem:nhbfpi}
Let $N(P)$ denote the neighbourhood of $P$. $N(P)=\biguplus_{\pi \in \Pi}\beta(\pi)$.
\end{lemma}
\begin{proof}
For a path $Q \in N(P)$, $\alpha(P,Q) \neq NULL$, and thus $Q$ is an element of $\beta(\pi)$ for some $\pi \in Pi$.
By Lemma~\ref{lem:disjointF},  for each pair of distinct $\pi,\pi'$, $\beta(\pi) \cap \beta(\pi') = \phi$. Thus $N(P)=\biguplus_{\pi \in \Pi}\beta(\pi)$. 
\end{proof}

\noindent

\section{The Succinct Data Structure}
\label{sec:succrep}
In this section, we present the construction of the succinct representation followed by the implementation of the queries. The input to our construction procedure is $(T, \mathcal{P})$ obtained from the path graph $G$ using Gavril's method~\cite{Gavril_path} where $T$ is the clique tree and $\mathcal{P}=\{P_1,\ldots,P_n\}$ is the set of paths in it such that the paths correspond to vertices of $G$ and have a non-empty intersection of their vertex sets if and only if the corresponding vertices are adjacent. The construction procedure starts by pre-processing the clique tree as explained in the previous section followed by storing it and the paths in a space efficient manner. We demonstrate a polynomial time construction mechanism without worrying about the most optimal way.

\begin{figure}[ht]
\centering
\includegraphics[width=1\textwidth]{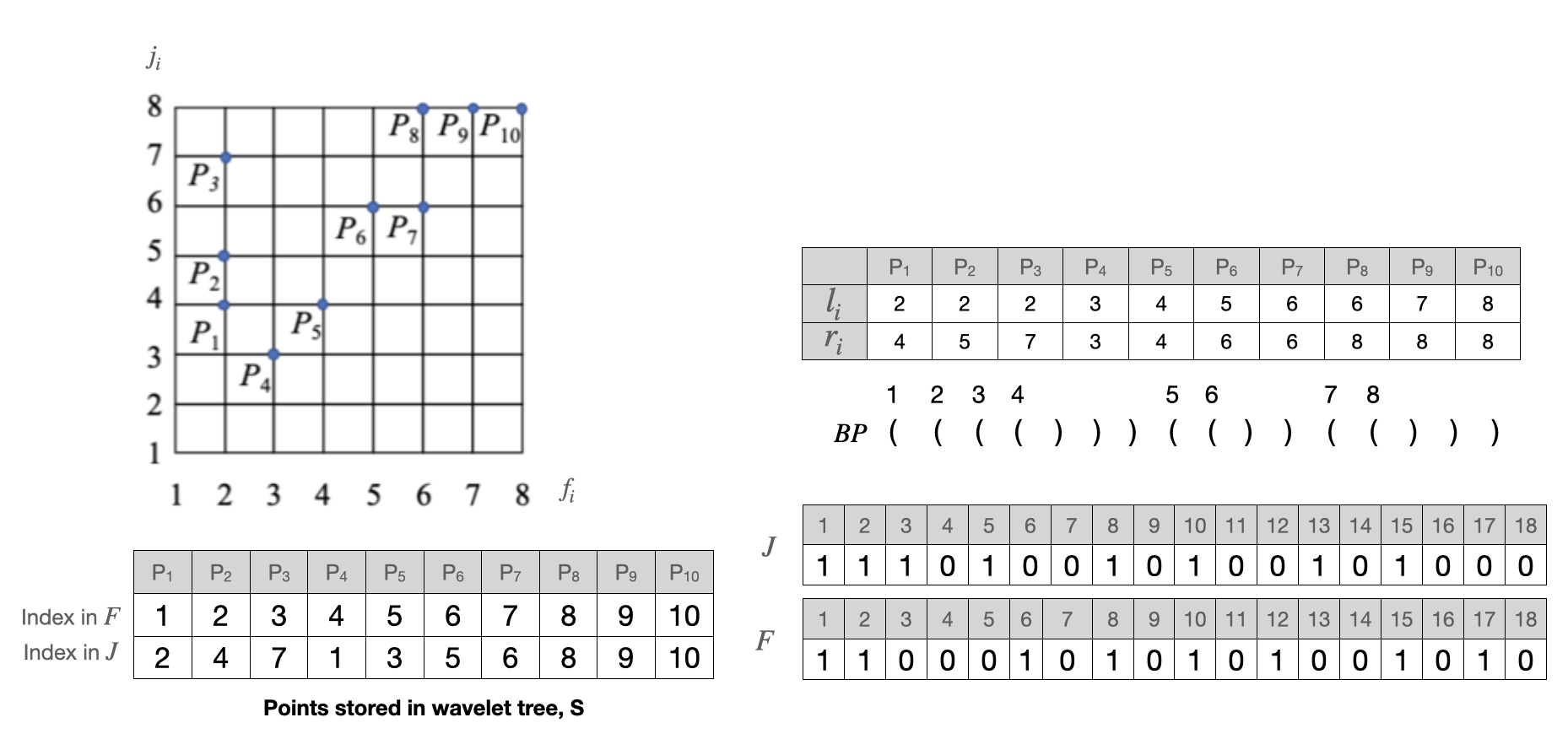}
\caption{Succinct representation for the  path graph $G$ of Figure~\ref{fig:pathgraph}. Top  left diagram shows the paths as points in a grid.  Bottom left shows the path aliases stored in the wavelet tree. The right side shows the $l_i$ and $r_i$ values of paths, the balanced parentheses representation of the pre-processed clique tree along with the $F$ and $J$ data structures.}
\label{fig:sucrep}
\end{figure}

\subsection{Construction of the Succinct Data Structure}
\label{sec:sdsconstruction}
Our succinct data structure for path graphs has two main  parts - the clique tree $T$ and the paths $\mathcal{P}=\{P_1,\ldots,P_n\}$ in it. The construction uses other compact data structures~\cite{Navarro} which are of the types: ordinal tree, bit vector, wavelet tree,  and array of sorted integers. In the next two sections we will explain the construction and storage of the clique tree and the paths in it.

\subsubsection{Storing the Clique Tree}
In this section we explain how the clique tree $T$ and its BP representation is stored succinctly.

\noindent
{\bf Clique tree $T$.} By Remark~\ref{rem:maxn}, the clique tree has at most $n$ nodes and is an ordinal tree. It is stored using $2n+o(n)$ bits using the data structure of 
Lemma~\ref{lem:2ntree}.  

\noindent
{\bf Bit-vector $BP$.} The balanced parentheses representation of $T$ is stored using the data structure of Lemma~\ref{lem:bstr} in bit-vector $BP$ using $2n+o(n)$ bits. In $BP$ the open and close parenthesis are represented by  bit  1 and  0, respectively. For every node $v$ in $T$ there exists two indices $i$ and  $j$ in $BP$ where $i<j$ such that $BP[i]=1$ and $BP[j]=0$. For some $1 \le i \le 2n$, if $BP[i]=1$ and $BP[i-1]=0$ then they represent the open and close parenthesis of nodes $v,u  \in V(T)$ that have a common parent $w$. Since $T$ is ordinal, in the order of children of $w$, $u$ comes immediately before $v$ and it is called $v$'s \textit{previous sibling}. The following three methods are supported by $BP$:
\begin{enumerate}
    \itemsep0em 
    \item $\texttt{getPreorder}(i)$: For $1 \le i \le 2n$ such that $BP[i]=1$, returns the pre-order label of the node which has its open parenthesis at $i$ in $BP$. It  is implemented by $\texttt{rank}(BP,1,i)$ for $i \ne 0$ and for $i=0$ it returns $1$. 
    \item $\texttt{getIndex}(v)$: Returns the index of the open parenthesis of $v \in V(T)$ in $BP$. It is implemented by $\texttt{select}(BP,1,v)$. 
    \item $\texttt{getHPStartNode}(v)$: Returns  the start node of heavy path $\pi$ that contains $v \in V(T)$ in constant time. If $v$ is not the first child, that is, it is adjacent to its parent by a light edge, then $v$ itself is returned else the method returns $\texttt{getPreorder}(\texttt{select}(BP,0,\texttt{rank}(BP,0,\texttt{getIndex}(v)))+1)$.
\end{enumerate}
\noindent
\begin{lemma}
\label{lem:getHPStartNode}
For $v \in V(T)$, ${\normalfont\texttt{getHPStartNode}}(v)$ returns in constant time the starting node of heavy path $\pi$ that contains $v$.
\end{lemma}
\begin{proof}
We need to show that the method $\texttt{getHPStartNode}$ as implemented above indeed obtains the  start node of $\pi$ in constant time. As we use constant time methods of Lemma~\ref{lem:bstr}, $\texttt{getHPStartNode}$ also completes in constant time. To show that $\texttt{getHPStartNode}$ returns the starting node of $\pi$ we consider the two cases depending on $v$:
\begin{enumerate}
    \itemsep0em 
    \item \textit{When $v$ is the root node of $T$ i.e. $v=1$}: $\texttt{getIndex}(v)$ returns 1 when $v=1$ and $\texttt{select}(BP,0,\texttt{rank}(BP,0,1))$ returns 0. Further, $\texttt{getPreorder}$ on input 1 returns 1. Thus, $\texttt{getHPStartNode}$ returns $v$ when input $v$ is the root  node, as it is the start node of $\pi$.
    \item \textit{When $v$ is not the root  node of $T$ i.e. $v\ne 1$}: Let $BP[i]$ be the open parenthesis of $v$ and $x$ denote the starting node of $\pi$. Also, let $BP[j]$ be the closing parenthesis of the previous sibling of $x$ in $BP$. Since $BP[j+1]$ is the open parenthesis of $x$, the length of path from $v$ to $x$ is $l=i-j-1$. The base case is when $l=0$ that is when $v$ is the starting node of $\pi$. In this case, $\texttt{getHPStartNode}$ returns $v$ itself. When $l>0$,  $\texttt{getIndex}(v)$ returns the position $i$ of the open parenthesis of $v$ in $BP$. $\texttt{select}(BP,0,\texttt{rank}(BP,0,i))$ returns $j$, the index of the closing parenthesis of the previous sibling of $x$. $\texttt{getPreorder}(j+1)$ thus returns the start node of $\pi$ correctly. 
\end{enumerate}
\end{proof}

\subsubsection{Storing the Paths $P_1,\ldots,P_n$}
To store path $P_i, 1 \le i \le n$ we need to store its starting  node $l_i$  and its ending  node $r_i$ in a space efficient way. Let $M=(M_1,\ldots,M_n)$ and $N=(N_1,\ldots,N_n)$ be the sequence of starting and ending nodes of paths sorted in non-decreasing order, respectively. For $1 \le i \le n$, $M[i]$ is the starting node of path $P_i$. On the other hand, for $1 \le i \le n$, $N[i]$ is the $i-$th ending node in the non-decreasing sorted order of ending nodes.

\noindent
{\bf Bit-vectors $F$ and $J$.}  $M$ and $N$ are stored in data structures $F$ and $J$, respectively, using the data structure of Lemma~\ref{lem:ssn} taking $2n+o(n)$ bits each.

\begin{proposition}
For $1 \le i \le n, {\normalfont\texttt{accessNS}}(F,i)$ returns $M[i]$ stored in $F$ in constant time.
\end{proposition}

\begin{proposition}
For $1 \le i \le n, {\normalfont\texttt{accessNS}}(J,i)$ returns $N[i]$ stored in $J$ in constant time.
\end{proposition}

\noindent
$F$ supports the following useful function too:
\begin{itemize}
    \item $\texttt{getPathCount}(d)$: Returns the number of paths that start at node $d \in V(T)$. When $\texttt{select}(F,1,d)$ is well defined and $F[\texttt{select}(F,1,d)+1]=0$, the count is obtained using the expression $\texttt{rank}(F,0,\texttt{select}(F,1,d+1)) - \texttt{rank}(F,0,\texttt{select}(F,1,d))$. In all other cases the function returns 0.
\end{itemize}
\begin{lemma}
For $x \in V(T)$, method ${\normalfont\texttt{getPathCount}}(x)$ returns $|\{P_i: P_i \in \mathcal{P}, l_i=x\}|$ where $l_i$ is the starting node of path $P_i$ in constant time. 
\end{lemma}
\begin{proof}
Let input $x \in [n]$ be a valid $l_i$ value of some path in $\mathcal{P}$, that is, $\texttt{select}(F,1,d)$ is well defined and $F[\texttt{select}(F,1,d)+1]=0$. If the $l_i$ value $x$ is repeating in $F$ then there will be a contiguous sequence of two or more 0's between the $x$-th 1 and the $x+1-$th 1.  Let $n_1$ be the number of 0's  before the $x+1-$st 1. It can be obtained using the expression $\texttt{rank}(F,0,\texttt{select}(F,1,d+1))$. Let $n_2$ be the number of 0's before the $x-$th 1. $n_2$ can be obtained using the expression $\texttt{rank}(F,0,\texttt{select}(F,1,d))$. The number of times $x$ is repeating is $n_1-n_2$. As per Lemma~\ref{lem:bstr} all these operations  can be done in constant time.
\end{proof}

\noindent
Next, we need to associate the path $P_i$ with its starting and ending nodes stored in $F$ and $J$. Starting node of $P_i$ is available directly from $F$ using $\texttt{accessNS}(F,i)$ whereas to get the ending node we need to associate it with its ending node's position in $J$. This association is established using a wavelet tree as described below.

\noindent
{\bf Wavelet tree $S$.} For each path $P_i, 1 \le i \le n$ we assign the tuple $(f_i,j_i)$ where $f_i$ and $j_i$ are indices of the $l_i$ and $r_i$ values in $M$ and $N$ respectively. Since paths are numbered based on the non-decreasing order of their starting nodes, $i = f_i$. In other words, $(f_i,j_i)$ acts as an alias for path $P_i=(l_i,r_i)$ and they have the following property.

\begin{lemma}
\label{lem:distpaths}
Let $\mathcal{P}'=\{(f_1,j_1),\ldots,(f_n,j_n)\}$  be the set of  aliases of paths in $\mathcal{P}$. The following are true:
\begin{enumerate}
    \itemsep0em 
    \item $1 \le k \ne l \le n, f_k \ne f_l$  and $j_k \ne j_l$
    \item $\mathcal{P}'$ can be stored using the wavelet tree $S$ using $n \log n + o(n \log n)$ bits of space such that $S$ supports the following method:
    \begin{enumerate}
        \itemsep0em 
        \item ${\normalfont\texttt{accessWT}}(S,i)$: For $i \in [n]$, returns $j_i$ in $O(\log n)$ time where $N[j_i]$ is the ending node of path $P_i$.
        \item ${\normalfont\texttt{searchWT}}(S,[i_1,i_2],[j_1,j_2])$: For $i_1,i_2,j_1,j_2 \in [n]$, returns $\{i \mid P_i \in \mathcal{P}', f_i \in [i_1,i_2] \text{ and } j_i  \in [j_1,j_2]\}$ in $O(\log n)$ time per path.
        \item ${\normalfont\texttt{countWT}}(S,[i_1,i_2],[j_1,j_2])$: For $i_1,i_2,j_1,j_2 \in [n]$, returns $\mid\{i \mid P_i \in \mathcal{P}', f_i \in [i_1,i_2] \text{ and } j_i  \in [j_1,j_2]\} \mid$ in $O(\log n)$ time.
    \end{enumerate}
\end{enumerate}
\end{lemma}
\begin{proof}
The proof is as follows:
\begin{enumerate}
    \itemsep0em 
    \item $P_i, 1 \le i \le n$ has its starting and ending nodes stored at unique indices $f_i$ and $j_i$ in $M$ and $N$, respectively. This ensures that for $1 \le k \ne l \le n, f_k \ne f_l$  and $j_k \ne j_l$.
    \item For path $P_i, 1 \le i \le n$, the wavelet tree of Lemma~\ref{lem:rsds} stores $(f_i,j_i)$, where $M[i]$ and $N[j_i]$ are the starting and ending nodes of $P_i$.  $\texttt{accessWT}$, $\texttt{searchWT}$, and $\texttt{countWT}$ functions can be directly delegated to the $\texttt{access}$, $\texttt{search}$, and $\texttt{count}$ functions of the wavelet tree of Lemma~\ref{lem:rsds}. The time complexities also  follow from Lemma~\ref{lem:rsds}.
\end{enumerate}

\end{proof}

\noindent
{\bf Function \texttt{pathep}.} Given a path index  $i, 1 \le i \le n,$ we can now obtain its $l_i$ and $r_i$ values using the method $\texttt{pathep}$. The method takes $i$ as input and returns $(l_i,r_i)$ in $O(\log n)$ time as follows.
\begin{enumerate}
    \itemsep0em 
    \item $l_i=\texttt{accessNS}(F,i)$.
    \item $r_i = \texttt{accessNS}(J,\texttt{accessWT}(S,i))$.
\end{enumerate}

\begin{lemma}
\label{lem:pathep}
For $1 \le i \le n$, ${\normalfont\texttt{pathep}}(i)$  returns $(l_i,r_i)$ of $P_i$ in $O(\log n)$ time.
\end{lemma}
\begin{proof}
First we show that $\texttt{pathep}(i)$ returns the $l_i$ value of path $i$ correctly. The $l_i$ value of $i$ is the number of  1's before the $i-$th 0 in $F$ which is obtained by $\texttt{accessNS}(F,i)$. Now, we show that the correct $r_i$ value is returned by $\texttt{pathep}(i)$. To get the $r_i$ value which is stored in $J$ we have to  get the index $j$ of path $i$ in $J$. This can be obtained by querying $S$. We obtain the $r_i$ value from $J$ by $\texttt{accessNS}(J,\texttt{accessWT}(S,i))$. Since $\texttt{accessNS}$ takes constant time as per Lemma~\ref{lem:ssn} and $\texttt{accessWT}$ takes $O(\log n)$ time as per Lemma~\ref{lem:distpaths}, the total time  taken is $O(\log n)$ time. 
\end{proof}

\noindent
{\bf Function $\texttt{maprange}^F$/$\texttt{maprange}^J$.} Given range $[l,l']$ of starting nodes of paths as input, $\texttt{maprange}^F$ outputs the range $[j,j']$ where $j$ is the first index in $M$ such that $M[j] \ge l$  and $j'$ is the last index in $M$  such that  $M[j'] \le l'$.  
\begin{enumerate}
    \itemsep0em 
    \item $j$ is obtained using the  expression  $\texttt{rank}(F,0,\texttt{select}(F,1,l))+1$ that returns the index in $M$ of the first occurrence of $l$ or a value greater than $l$ but less than or equal to $l'$. 
    \item To obtain $j'$ we use the following steps:
    \begin{enumerate}
        \itemsep0em 
        \item If $M[\texttt{rank}(F,1,\texttt{select}(F,1,l')+1)]=l'$ then return \\  $\texttt{rank}(F,0,\texttt{select}(F,1,l')+1) + \texttt{getPathCount}(l')-1$. In other words, if $l'$  is present in  $M$ then $j'$ is the index of the last $l'$ in  $M$. To account for the repeating $l'$ we add to the the first occurrence of $l'$ in $M$ one less than the number of times the $l'$ value repeats. 
        \item If $M[\texttt{rank}(F,1,\texttt{select}(F,1,l')+1)] \ne l'$ then return \\ $\texttt{rank}(F,0,\texttt{select}(F,1,l'))$. If $l'$  is not present then $M[j']$ is a value that is less than $l'$ but greater than or equal to $l$. 
    \end{enumerate}
    $M[z], z \in [n]$ can be obtained using $\texttt{accessNS}(F,z)$.
\end{enumerate}

\begin{lemma}
\label{lem:maprange}
Given a  range $[l,l']$ of starting nodes where $l,l' \in [n]$, ${\normalfont \texttt{maprange}^F}(l,l')$ returns the range $[j,j']$ in constant time where $j$ and $j'$ are the smallest and largest indices in $M$ such that $M[j] \ge l$ and $M[j'] \le l'$.
\end{lemma}
\begin{proof}
First we will show that $j$ is computed correctly by the expression \\ $\texttt{rank}(F,0,\texttt{select}(F,1,l))+1$. In the unary encoding in $F$, $\texttt{select}(F,1,l)$ identifies the position $i$ of the $l-$th 1. If $l$ is present in $F$ then $F[i+1]$ is a 0 else its a 1. If $F[i+1]=0$ then $\texttt{rank}(F,0,i)+1$ returns the index $j$ of $l$ in $M$. On the other hand, if $F[i+1]=1$ then let $k$ be the smallest number such  that $F[i+k]=0$. In this case, $\texttt{rank}(F,0,i)+1$ returns the smallest index $j$ in $M$ of $l''>l$. Next, we show that $j'$ is returned correctly. If $l'$ is present in $M$ then $\texttt{rank}(F,0,\texttt{select}(F,1,l')+1)$ returns the index $j''$ of the first $l'$ in $M$. The largest index in $M$ of $l'$ is obtained by adding $\texttt{getPathCount}(l')-1$ to $j''$. On the other hand, if $l'$ is not in $M$ then $\texttt{rank}(F,0,\texttt{select}(F,1,l'))$ returns the largest index $j$ in $M$ of $l'' < l'$.  By Lemma~\ref{lem:bstr}, $\texttt{rank}$ and $\texttt{select}$ can be completed in constant time. Also, by Lemma~\ref{lem:ssn}, $\texttt{accessNS}$ takes constant time. Thus, $\texttt{maprange}^F$ completes in constant time.
\end{proof}

\noindent
We have a similar function, $\texttt{maprange}^J$ for mapping the range $[r,r']$ of ending nodes of paths to range $[j,j']$ such that $j$ is the smallest index in $N$ such that $N[j] \ge r$ and $j'$ is the largest index in $N$ such that $N[j'] \le r'$.

\begin{lemma}
\label{lem:maprangeJ}
Given a  range $[r,r']$ of ending nodes where $r,r' \in [n]$, ${\normalfont \texttt{maprange}^J}(r,r')$ returns the range $[j,j']$ in constant time where $j$ and $j'$ are the smallest and largest indices in $N$ such that $N[j] \ge r$ and $N[j'] \le r'$.
\end{lemma}

\noindent
{\bf Bit-vector D.} Bit vector $D$ of size $n$ stores for each path a 1 if the path intersects with more than $\log n$ other paths else a 0. It supports the following function.
\begin{itemize}
    \item $\texttt{isLargeDegree}(i)$: Returns true if $D[i]=1$ else false in constant time.
\end{itemize}

\begin{lemma}
\label{lem:succinctds}
There exists an $n \log n + o(n \log n)$-bit succinct data structure for path graphs.
\end{lemma}
\begin{proof}
The space taken by the components of the succinct data structure for path graphs are as follows:
\begin{enumerate}
    \itemsep0em 
    \item The clique tree $T$ and its $BP$ representation takes $O(n)$ bits of space. This follows from Lemma~\ref{lem:2ntree} and~\ref{lem:bstr}.
    \item To store the end points of paths in $\mathcal{P}$ we have bit vectors $F$, $J$. From Lemma~\ref{lem:ssn} this also takes $O(n)$ bits.
    \item The wavelet tree stores the indices of paths in $M$ and $N$ and from Lemma~\ref{lem:distpaths} takes $n \log n + o(n \log n)$ bits.
    \item To improve  the degree query we  have the $n$ bit  vector $D$.
\end{enumerate}

The space complexity of the succinct representation is dominated by the space required for wavelet tree $S$. Thus, our representation takes $n \log n + o(n \log n)$ bits. This representation is succinct as it uses the permitted storage for succinct representation  of  interval graphs~\cite{HSSS} that is a proper sub-class of path graphs~\cite{agtpg}. 
\end{proof}

\subsection{Adjacency and Neighbourhood Queries}
\label{sec:supqueries}
In this section, we will present efficient implementations of adjacency and neighbourhood queries using the succinct representation as constructed in Section~\ref{sec:sdsconstruction}. In this section, as a consequence of Lemma~\ref{lem:succinctds}, the succinct representation for path graph $G$ is denote  as $(T, \mathcal{P})$.  Adjacency query, as will be shown in Lemma~\ref{lem:adjacency}, takes two path indices $i,j \in [n]$ and the succinct representation $(T, \mathcal{P})$ as input and returns true if the paths $P_i$ and $P_j$ have a non-empty intersection. The neighbourhood query, as will be shown in Lemma~\ref{lem:nhbrsds}, takes a single path index $i \in [n]$ and the succinct representation $(T, \mathcal{P})$ as input and returns the list of paths that have non-empty intersection with of the path $P_i$. The implementation of the queries depend on the following:
\begin{enumerate}
    \itemsep0em 
    \item Computing paths $P=(l,r)$ and $Q=(s,t)$ corresponding to $i$ and $j$, respectively using $\texttt{pathep}$ and $p=\texttt{lca}(l,r)$ in $O(\log  n)$ time.
    \item Computing $\Pi, k , \texttt{succ}(1,1)$ and $\texttt{succ}(1,2)$. From Lemma~\ref{lem:computePi}  that follows, Algorithm~\ref{alg:computePi} can do this in $O(\log n)$ time.
    \item Computing $\beta(\pi)$ for $\pi \in \Pi$. From Lemma~\ref{lem:computeBeta} that follows, $\beta(\pi)$ can be computed in $O(d \log n)$ time where $d$ is the  number of paths returned by $\beta(\pi)$.
\end{enumerate}


\begin{algorithm}[ht!]
\label{alg:computePi}
\caption{Given path $P=(l,r)$ as input, function $\texttt{compute}\Pi$ computes $\Pi$, $k$, $\texttt{succ}(1,1)$, and $\texttt{succ}(1,2)$. We assume that $\texttt{parent}(v)=0$ when $v$ is the root of the tree.}
\DontPrintSemicolon
    \SetKwFunction{FaddHSP}{addHSP}
    \SetKwFunction{FcomputePi}{compute$\Pi$}
    \SetKwFunction{FcomputePiH}{compute$\Pi$\_Helper}
    \SetKwProg{Fn}{Function}{:}{}
    \Fn{\FcomputePi{$l,r$}}{
        $p \leftarrow \texttt{lca}(l,r)$\;
        $k \leftarrow 0$\;
        $\Pi=\Pi'=\texttt{succ}(1,1)=\texttt{succ}(1,2)=\texttt{NULL}$\;
        \If{ $l \ne p$}{
            \FcomputePiH{$p,l,\Pi,k$} \;
            \FcomputePiH{$p,r,\Pi',k$} \;
            Add second entries of $\Pi$ and $\Pi'$ as $\texttt{succ}(1,1)$ and $\texttt{succ}(1,2)$ respectively \;
            Concatenate $\Pi'$  to $\Pi$ preserving the order $\prec$ \;
        }
        \Else{
            \FcomputePiH{$l,r,\Pi,k$} \; 
            \If{first entry in $\Pi$ has equal starting and ending  nodes}{
	           If starting node of the first entry in $\Pi$ is the parent of starting node of the second entry
	           then $\texttt{succ}(1,1)$  is the second entry in $\Pi$ and $\texttt{succ}(1,2)=\texttt{NULL}$ \;
            }
            \Else{
	           If ending node of the first entry in $\Pi$ is the parent of the starting node of the second  entry
	           then $\texttt{succ}(1,1)$ is the second entry in $\Pi$ and $\texttt{succ}(1,2)=\texttt{NULL}$ \;
            }
        }
    }
    \Fn{\FcomputePiH{$l,r,\Pi,k$}}{
	    \If {$l>r$} \KwRet \; 
	    $p=l$\;
	    $u \leftarrow  \texttt{getHPStartNode}(r)$ \;

       Increment $k$\;
       \If {$u >= p$}{  
        Add $(u,r)$ to beginning of $\Pi$ \;
        \FcomputePiH{$p,\texttt{parent}(u),\Pi,k$} \;
       }
      \Else{
        Add $(p,r)$ to beginning of $\Pi$ \;	
      }
    }      
\end{algorithm}

\begin{lemma}
\label{lem:computePi}
Given a path $P=(l,r)$, ${\normalfont\texttt{compute$\Pi$}}(l,r)$ computes $\Pi$, $k$, ${\normalfont\texttt{succ}}(1,1)$, and ${\normalfont\texttt{succ}}(1,2)$ for it in $O(\log n)$ time.
\end{lemma}
\begin{proof} 
First we show  that $\texttt{compute}\Pi$ of Algorithm~\ref{alg:computePi} computes the heavy sub-paths of $P$ as in Lemma~\ref{lem:pathpartition}. Function $\texttt{compute}\Pi$ depends on the function $\texttt{compute$\Pi$\_Helper}$ to compute the heavy sub-paths. Paths are of two types depending on whether the lca is same as its starting node. Based on this distinction different steps are executed in the function $\texttt{compute}\Pi$; see Line 5 of Algorithm~\ref{alg:computePi}. 
\begin{enumerate}
    \itemsep0em 
    \item Type 1 paths: If lca of $P$ is not equal to $l$ then the heavy sub-paths that comprise the sub-path from $p$ to $l$ are computed first. This is followed by computing the heavy sub-paths that comprise the sub-path from $p$ to $r$. This is done using the function $\texttt{compute$\Pi$\_Helper}$ as shown in Line 6 and 7 of Algorithm~\ref{alg:computePi}. $\texttt{compute$\Pi$\_Helper}$$(p,l,\Pi,k)$ computes the heavy sub-paths recursively till $\pi_1$; see Line 6 of Algorithm~\ref{alg:computePi}. Starting at $l$, the starting node of the heavy sub-path to which it belongs is  obtained by using $\texttt{getHPStartNode}$; see Line 21 of Algorithm~\ref{alg:computePi}. The set of heavy sub-paths are computed in this manner till $p$ is reached; see Line 23 to 25 of Algorithm~\ref{alg:computePi}. Similar steps are performed for $\texttt{compute$\Pi$\_Helper}$$(p,r,\Pi,k)$; see Line 7 of Algorithm~\ref{alg:computePi}. This gives us the end points of the heavy sub-paths of $P$.
    \item Type 2 paths: If lca of $P$ is equal to $l$ then the heavy sub-paths comprising the only sub-path from $l=p$ to $r$ is computed using the function $\texttt{compute$\Pi$\_Helper}$ as shown in Line 11 of Algorithm~\ref{alg:computePi}. Heavy sub-paths for type 1 paths are also computed just as heavy sub-paths for type 1; see Line 11 to 15 in Algorithm~\ref{alg:computePi}.
\end{enumerate}

It takes $O(\log n)$ time to compute heavy sub-paths as there are $O(\log n)$ light  edges (or heavy sub-paths) as per Lemma~\ref{lem:pathpartition}  and as per Lemma~\ref{lem:getHPStartNode}, $\texttt{getHPStartNode}$ takes constant time. From Lemma~\ref{lem:2ntree}, $\texttt{lca}$ and $\texttt{parent}$ also take constant time. Since function $\texttt{compute}\Pi$ calls $\texttt{compute$\Pi$\_Helper}$ only a constant number of times, the complexity of the $\texttt{compute}\Pi$ function  is also $O(\log n)$. 
\end{proof}

\noindent
From Lemma~\ref{lem:nhbfpi}, we know that the neighbourhood query depends on computing $\beta(\pi)$ for all $\pi \in \Pi$. Next, we show that $\beta(\pi)$  can be computed in $O(d_{\pi} \log n)$ time where $d_{\pi}$ is $|\beta(\pi)=\{Q| Q \in \mathcal{P} \text{ and  } \alpha(P,Q)=\pi \}|$. By an abuse of terminology, $d_{\pi}$ is called the \textit{degree} of $\pi$.  

\begin{lemma}
\label{lem:computeBeta}
Given index $i$ of $\pi_i \in \Pi$, there exists  a function \\ ${\normalfont \texttt{compute}\beta(i, \Pi, \texttt{succ}(1,1),\texttt{succ}(1,2))}$ that returns $\beta(\pi_i)=\{Q| Q \in \mathcal{P} \text{ and } \alpha(P,Q)=\pi_i  \}$ in $O(d_{\pi_i} \log n)$  time  where $d_{\pi_i}$ is the degree of $\pi_i$.  
\end{lemma}
\begin{proof}
First we will show that there exists a function $\texttt{compute}\beta(i, \Pi, \texttt{succ}(1,1),\texttt{succ}(1,2))$ that computes  $\beta(\pi_i)$ correctly. The high level steps of function $\texttt{compute}\beta$ are as follows.
\begin{enumerate}
    \itemsep0em 
    \item Compute the interval ranges $R_j, 1 \le j \le 4,$ of $\pi_i$ using its end points and its successor stored in  $\Pi$. If $i=1$ then the successors are directly  available in the input else it can be obtained from $\Pi$  as follows. For $i \ne 1$, it is $\pi_{i+1}$ unless $\pi_{i+1}=\texttt{succ}(1,2)$ or $i=k$ where $k$ is the number of heavy sub-paths in $\Pi$. 
    \item The next step is to identify all $Q \in \mathcal{P}$ that satisfy $\alpha(P,Q)=\pi_i$. The ranges of  starting  and ending nodes of such paths can be obtained from the  conditions of Lemma~\ref{lem:FFconditions}. Using these ranges the paths can be retrieved by issuing orthogonal range search queries on wavelet  tree $S$ of Lemma~\ref{lem:distpaths}. The ranges corresponding to first two conditions of Lemma~\ref{lem:FFconditions} can be directly obtained. For the last two conditions we use Lemma~\ref{lem:lca_range}.
    \item $\texttt{searchWT}$ from Lemma~\ref{lem:distpaths} is used to perform the orthogonal range search on wavelet tree $S$.
\end{enumerate}
As there are only four interval ranges for $\pi_i$ and from  Lemma~\ref{lem:2ntree}, $\texttt{rmost\_leaf}$ takes constant time, the interval ranges of $\pi_i$ can be computed in constant time. From these interval ranges the ranges for orthogonal range search can be obtained using Lemma~\ref{lem:FFconditions}. This can be done in constant time as from Lemma~\ref{lem:2ntree}, $\texttt{lca}$ takes constant time. $\texttt{searchWT}$ takes $O(d \log n)$ time per range  query where $d$ is the number of paths in $\mathcal{P}$ with  starting and ending nodes in the input range. There is no over counting of paths between range search queries as no path satisfies more than one condition due to Lemma~\ref{lem:FFconditions}. Since there are only four  orthogonal range queries to be issued for any heavy sub-path, $\texttt{compute}\beta$ completes in $O(d_{\pi_i} \log n)$ time. 
\end{proof}

\noindent
{\bf Adjacency query in  $O(\log n)$ time.} Given indices of paths $i,j \in [n]$ and $(T, \mathcal{P})$ as input, adjacency query returns  $\texttt{true}$ if paths corresponding to $i$ and $j$, namely $P$ and $Q$, have a non-empty intersection in $T$. Adjacency  of paths with indices $i$ and $j$ can be checked as shown  in Algorithm~\ref{alg:adjacency}. We have the following lemma.

\begin{algorithm}[ht!]
\label{alg:adjacency}
\caption{Given two path indices $i,j \in [n]$, the function $\texttt{adjacency}$ checks if the paths corresponding to them have  a non-empty intersection. }
\DontPrintSemicolon
    \SetKwFunction{Fadjacency}{adjacency}

    \SetKwProg{Fn}{Function}{:}{}
    \Fn{\Fadjacency{$i,j$}}{
        Obtain paths $P=\texttt{pathep}(i)$ and $Q=\texttt{pathep}(j)$. Let $P=(l,r)$ and $Q=(s,t)$.\;
        Initialize $k \leftarrow 0$ and $\Pi$ to empty \;
        $(\Pi,k,\texttt{succ}(1,1),\texttt{succ}(1,2)) \leftarrow\texttt{compute}\Pi(l,r)$ \;
        For each $1 \le i \le k$ return true if $\texttt{check}\alpha(i,Q,\Pi,\texttt{succ}(1,1),\texttt{succ}(1,2))$ of Lemma~\ref{lem:checkalpha} returns true\;
    }    
\end{algorithm}

\begin{lemma}
\label{lem:adjacency}
Given two path indices $i,j \in [n]$ and $(T,\mathcal{P})$ as  input, the function ${\normalfont \texttt{adjacency}(i,j)}$ checks if paths corresponding to $i$ and $j$ have a non-empty intersection in $O(\log n)$ time.
\end{lemma}
\begin{proof}
By definition, if $\alpha(P,Q)=\pi$ for $\pi \in \Pi$ then $Q$ and $P$ are adjacent. The  existence of such a  heavy sub-path can be tested as  shown  in Line 5 of  Algorithm~\ref{alg:adjacency}. Paths $P=(l,r)$ and $Q=(s,t)$ corresponding to $i$ and $j$, respectively, can be obtained in $O(\log n)$ time using $\texttt{pathep}$ due to Lemma~\ref{lem:pathep}. By Lemma~\ref{lem:computePi}, $\Pi, k,$  and the successors of $\pi_1$ can be computed in $O(\log n)$ time. For each heavy sub-path $\pi \in \Pi$, the conditions of Lemma~\ref{lem:FFconditions} can be checked in constant time using $\texttt{check}\alpha$ of Lemma~\ref{lem:checkalpha}. Also, from Lemma~\ref{lem:2ntree}, $\texttt{rmost\_leaf}$ can be computed in constant time.  Since by Lemma~\ref{lem:pathpartition}$, \Pi$ contains at most $O(\log n)$ heavy sub-paths, the total time taken is $O(\log n)$.
\end{proof}

\noindent
{\bf Neighbourhood query.}  Given a path index $i \in [n]$ and $(T, \mathcal{P})$, the  neighbourhood query returns the neighbours of path $P$ corresponding to index $i$; see Lemma~\ref{lem:nhbfpi} for definition of neighbours of  a path. Let $N(P)$ be initialized to empty. $N(P)$ can be obtained as shown  in Algorithm~\ref{alg:neighbourhood}. We call $|N(P)|$ the \textit{degree} of  $P$. We have the following lemma.

\begin{algorithm}[ht!]
\label{alg:neighbourhood}
\caption{Given path index $i$, the function $\texttt{neighbourhood}$ enumerates the paths that have non-empty intersection with $P$. }
\DontPrintSemicolon
    \SetKwFunction{Fneighbourhood}{neighbourhood}

    \SetKwProg{Fn}{Function}{:}{}
    \Fn{\Fneighbourhood{$i$}}{
        Obtain end points $(l,r)$ of $P$ using $\texttt{pathep}(i)$ \;
        Initialize $k \leftarrow 0$ and $\Pi$ to empty \;
        Compute $(\Pi,k,\texttt{succ}(1,1),\texttt{succ}(1,2))$ for $P$ using the $\texttt{compute}\Pi(l,r)$ function\;
        For each $\pi \in \Pi$ add $\texttt{compute}\beta(\pi)$ of Lemma~\ref{lem:computeBeta} to $N(P)$. 
    }    
\end{algorithm}

\begin{lemma}
\label{lem:nhbrsds}
Given path index $i \in [n]$ of path $P \in \mathcal{P}$  and $(T,\mathcal{P})$ as  input, the function ${\normalfont\texttt{neighbourhood}}(i)$ returns the set of neighbours of $P$ in $O(d_P \log n)$ time where $d_P$ is the degree of $P$.   
\end{lemma}
\begin{proof}
From Lemma~\ref{lem:nhbfpi}, the neighbours of $P$ are the paths in $\biguplus\limits_{i=1}^k \beta(\pi_i)$. The end points of $P$ can be obtained in $O(\log n)$ time using $\texttt{pathep}$ due to Lemma~\ref{lem:pathep}. From Lemma~\ref{lem:computePi}, we know that $\texttt{compute$\Pi$}$ takes $O(\log n)$ time and from Lemma~\ref{lem:computeBeta}, we know that $\texttt{compute}\beta$ takes $O(d_{\pi} \log n)$ time for each $\pi \in \Pi$ where $d_{\pi}$ is the number of paths that $\alpha$ maps to $\pi$. The time taken by $\texttt{neighbourhood}$ is sum of the time taken by $\texttt{pathep},$ $\texttt{compute}\Pi$ and at most $k$ iterations of $\texttt{compute}\beta$. Since by Lemma~\ref{lem:nhbfpi}, we know that none of the neighbours are over-counted the total time  taken is $O(d_P \log n)$  where $d_P=\sum\limits_{i=1}^k d_{\pi_i}$ where $d_P$ is  the degree of path $P$. 
\end{proof}

\noindent
{\bf Degree query.} Degree of path $P$ can be obtained by two different methods depending on the degree of the path. We use a bit vector $D$ as described in Section~\ref{sec:sdsconstruction}. We have two methods for computing degree of $P$ with index $i$ depending on $\texttt{isLargeDegree}(i)$.
\begin{enumerate}
    \itemsep0em 
    \item $\texttt{isLargeDegree}(i)$ is true: We modify Algorithm~\ref{alg:neighbourhood} for $\texttt{neighbourhood}$ to return the count of the orthogonal range search instead of the paths by  using $\texttt{countWT}$ of Lemma~\ref{lem:distpaths} instead of  $\texttt{searchWT}$. 
    \item $\texttt{isLargeDegree}(i)$ is false: We run the Algorithm~\ref{alg:neighbourhood} for $\texttt{neighbourhood}$ without modification and count the number of paths returned. 
\end{enumerate}

\begin{lemma}
\label{lem:degsucc}
Given path index $i \in [n]$ of path $P \in \mathcal{P}$  and $(T,\mathcal{P})$ as  input, the function ${\normalfont\texttt{degree}}(i)$ returns the degree of $P$ in $\min\{O(\log^2 n), O(d_P \log n)\}$ time where $d_P$ is the degree of $P$.    
\end{lemma}
\begin{proof}
As described above, two different methods are used depending on  whether $\texttt{isLargeDegree}(i)$ is true or not. Thus, we have the following two cases:
\begin{enumerate}
    \itemsep0em 
    \item $\texttt{isLargeDegree}(i)$ is true: $\texttt{countWT}$ of Lemma~\ref{lem:distpaths} takes $O(\log n)$ time. Since there are $O(\log n)$ heavy sub-paths as per Lemma~\ref{lem:pathpartition}, the total time is $O(\log^2 n)$. 
    \item $\texttt{isLargeDegree}(i)$ is false: By Lemma~\ref{lem:nhbrsds}, Algorithm~\ref{alg:neighbourhood} takes $O(d_P \log n)$ time. Thus, degree also takes $O(d_P \log n)$ time. 
\end{enumerate}
Since we run only one of the two depending on which is better, the  time taken by degree query is $\min\{O(\log^2 n), O(d_P \log n)\}$. 
\end{proof}

\noindent
\begin{proof}[Proof of Theorem \ref{thm:succthm}]
Lemma~\ref{lem:succinctds} shows that there exists a succinct  representation for path graphs that takes $n \log n + o(n \log n)$ bits. Given this representation as  input, Lemma~\ref{lem:adjacency} shows that adjacency between  vertices can be checked in $O(\log  n)$ time. Similarly, given this representation as input Lemma~\ref{lem:nhbrsds} and~\ref{lem:degsucc} show that for vertex $u \in V(G)$ with degree $d_u$, neighbourhood and degree queries are supported in $O(d_u \log n)$ and $\min\{O(\log^2 n), O(d_u \log n)\}$ time, respectively. Hence, Theorem~\ref{thm:succthm}.
\end{proof}
\noindent

\section{The Space-Efficient Data Structure}
\label{sec:seds}
We present an $O(n \log^2 n)$-bit space-efficient representation for path graphs that supports faster  adjacency and degree queries in comparison to the succinct representation presented in Section \ref{sec:succrep}.  The approach we take is to represent a path graph using the succinct data structure for interval graphs due to Acan et al.~\cite{HSSS}. To represent the path graph using the interval graph representation in \cite{HSSS} we end up having multiple {\em copies} of each vertex, and the adjacency between vertices could be witnessed in different interval graphs in our transformation.  Our data structure stores these interval graphs using the representation of \cite{HSSS}, along with an additional table to keep track of the copies of the vertices and edges.  
This transformation has an interesting contrast to the succinct data structure in Section \ref{sec:succrep}; there the path graph is represented using the clique tree and the adjacency queries are transformed to range queries.  \\

\noindent
The path graph $G$ is presented as $(T,\mathcal{P})$, where $T$ is a clique tree of $G$ and $\mathcal{P}=\{P_1,\ldots,P_n\}$ is the set of paths in $T$.
Consider the heavy path tree $\mathcal{T}$ of $T$.   Let $\mathcal{H}$ denote the set of heavy paths of $T$. \\
{\em Convention:} Let $M$ denote $|V(T)|$. It follows from Remark~\ref{rem:maxn} that $M$ and $|\mathcal{H}|$ are at  most $n$. $P_v$ denotes the path in $\mathcal{P}$ corresponding to vertex $v \in V(G)$. For a node $w \in \mathcal{T}$, we use $H_w$ to denote the heavy path in $T$  associated with the node $w$.  The level number of a node in ${\cal T}$ is one more than the number of edges on the path to it from the root; thus the level number of the root is 1.    $K$ denotes the number of levels in $\mathcal{T}$ and  level $l$ consists of the heavy paths which are at that level in ${\cal T}$.  
From Lemma~\ref{lem:hpd_prop}, $\mathcal{T}$ has at most $\lceil \log n \rceil$ levels and each path $P$ in $\mathcal{T}$ has at most $2 \lceil \log n \rceil$ edges.

\begin{lemma}
\label{lem:allinters}
For any $P,Q \in \mathcal{P}$, there exists a node $v$ in $\mathcal{T}$ such that $\Phi^{-1}(v)$ has a non-empty intersection with the path $P \cap Q$. 
\end{lemma}
\begin{proof}
From Lemma~\ref{lem:vparti}, we  know that nodes of $T$ are partitioned among the nodes of $\mathcal{T}$. This implies nodes of $P \cap Q$ belong to some $v \in V(\mathcal{T})$. Thus, for some $v \in V(\mathcal{T})$, $\Phi^{-1}(v)$ intersects with the path $P \cap Q$.
\end{proof}


\noindent

\begin{lemma}
\label{lem:charinter}
For $u$ and $v$ in $V(G)$, $P_u$ and $P_v$ have a non-empty intersection in $T$ if and only if one of the following is  true:
\begin{enumerate}
    \itemsep0em 
    \item there is a light edge $\{w,w'\}$ in $\mathcal{T}$  such that $P_u$ and $P_v$ both intersect $H_w$ and $H_{w'}$
    \item there is exactly a node $w$ in $\mathcal{T}$  such that $P_u \cap H_w$ and $P_v \cap H_{w}$ have a non-empty intersection.
\end{enumerate}
\end{lemma}
\begin{proof}
If $V(P_u) \cap V(P_v) \ne \phi$ then there are two possibilities:
\begin{enumerate}
    \itemsep0em 
    \item there exist $t_1,t_2 \in V(P_u) \cap V(P_v)$ and $t_1 \in V(H_w)$ and $t_2 \in V(H_{w'})$ such that there exists a light edge $\{w,w'\}$ in $\mathcal{T}$.
    \item there exists only  one $H_w \in \mathcal{H}$ that contains all nodes in $V(P_u) \cap V(P_v)$  for  some $w \in V(\mathcal{T})$. In this  case, $P_u \cap H_w$ and $P_v \cap H_w$ have a non-empty intersection.
\end{enumerate}
Conversely, if $P_u$ and $P_v$ both intersect heavy paths $H_w$ and $H_{w'}$ where $w,w' \in V(\mathcal{T})$ then $P_u$ and $P_v$ share the light edge $\{w,w'\}$. Thus, they intersect in $T$. If $P_u \cap H_w$ and $P_v \cap H_w$ intersect then by definition $P_u$ and $P_v$ intersect in $T$. 
\end{proof}

\noindent
{\bf Interval graph associated with heavy path $H$.} For a heavy path $H$ associated with a node in $\mathcal{T}$ of level number $l$, $G_H$ is a graph whose vertices are defined as follows: for each $1 \leq i \leq n$, if $P_i \cap H \ne \phi$ then there is a vertex corresponding to $P_i \cap H$ in $V(G_H)$. Two vertices are adjacent in $G_H$ if the corresponding paths have a non-empty intersection, otherwise they are not adjacent.  
\begin{lemma}
\label{lem:hpig}
Let $H \in \mathcal{H}$.  Then $G_H$ is an interval graph.  
\end{lemma}
\begin{proof}
The vertices of $G_H$ correspond to paths in $\mathcal{P}$ that have non-empty intersection with $H$. Thus,  it follows that each vertex of $G_H$ corresponds to a sub-path of $H$, which is equivalently an interval in the set $\{1, 2, \ldots, |V(H)|\}$.  Thus, $G_H$ is an interval graph.
\end{proof}
\begin{lemma}
Let $P_u$ and $P_v$ be paths in $T$.  $P_u \cap P_v \neq \phi$ if and only if there exist a heavy path $H \in {\cal T}$ such that in the interval graph $G_H$, the vertices corresponding to $P_u \cap H$ and $P_v
\cap H$ are adjacent.  
\end{lemma}
\begin{proof}
The proof follows directly from Lemma \ref{lem:charinter}.
\end{proof}
\noindent
It follows from the above lemma that each edge in $G$ has a representative in the interval graph associated with at least one of the heavy paths in ${\cal T}$. 
Thus, it is natural to group all the interval graphs into the levels associated with the heavy paths in ${\cal T}$. \\
{\bf Interval graph associated with a level $l$ in ${\cal T}$.} 
Let $S_l$ be the set of nodes in $\mathcal{T}$ at level $l$. For each $l$, define $U_l=\{G_{H_w} \mid w \in S_l\}$.  In other words, $U_l$ is the collection of  interval graphs associated with each heavy path at level $l$.  Thus, the vertex set and edge set of $U_l$ is the union of vertex sets and edge sets of $G_{H_w}$ for all $w \in S_l$. 
Clearly, $U_l$ is an interval graph. We next show that the number of vertices in $U_l$ is at most twice the number of vertices in $G$, that is, at most twice the number of paths in ${\cal P}$.  

\begin{lemma}
\label{lem:2nvertices}
Let $P_v$ be a path in ${\cal P}$ and $l$ be a level number in ${\cal T}$. 
There exists at most two nodes $w$ and $w'$ at level $l$ of $\mathcal{T}$ such that $G_{H_w}$ and $G_{H_{w'}}$  have a vertex each corresponding to the paths $P_v \cap H_w$ and $P_v \cap H_{w'}$.  Therefore, the number of vertices in  the interval graph $U_l$ is at most $2n$. 
\end{lemma}
\begin{proof}
We know from Lemma \ref{lem:pathpartition} that the nodes of $P_v$ are partitioned into heavy sub-paths, each of which is contained in a heavy path.  
Since each heavy path corresponds to a node in ${\cal T}$, it follows that  $P_v$ naturally defines a path $P$ in ${\cal T}$.
From Proposition~\ref{prop:2node}, it follows that $P$ has at most two nodes in level $l$.  Thus,  for each level $l$, $P_v$ has a non-empty intersection with  at most two heavy paths whose nodes are at level $l$ in ${\cal T}$.  Consequently, $U_l$ has at most $2n$ vertices.
\end{proof}

\noindent
In Section~\ref{sec:sedconstruction}, we present the data structure to store the set of interval graphs $\{U_l \mid 1 \leq l \leq K\}$ and additional tables to respond to the adjacency and neighborhood queries.

\subsection{Construction of the Space Efficient Data Structure}
\label{sec:sedconstruction}
The main goal of this section is to prove the space complexity part of Theorem \ref{thm:nlog2n}.
Given the $(T, \mathcal{P})$ representation for a path graph $G$ with $n$ vertices, the space-efficient data structure is constructed by the following steps:
\begin{enumerate}
    \itemsep0em 
    \item Compute the heavy path decomposition of clique tree $T$ and the heavy path tree $\mathcal{T}$ from $T$ as guaranteed in Section~\ref{sec:hpd}.
    \item Construct the ordinal clique tree $T$ as explained in Section~\ref{sec:proprocessingT}.
    \item Store the set of interval graphs $\{U_l \mid 1 \leq l \leq K\}$ using~\cite{HSSS} and a table called PIT that stores, for every level $1 \le l \le K$, the labels of vertices in $G_l$ corresponding to paths in $\mathcal{P}$; by Lemma~\ref{lem:2nvertices} there are two labels per path at a level $l$. 
\end{enumerate}

\noindent
 The construction of the data structure takes polynomial time and implementation details are left out. The components of the space-efficient representation are as follows.

\noindent
{\bf Array, $F$.} This is a one dimensional array of length $M$. $F[a]=i$ where $i$ is the index, in $\prec_{{\cal H}}$, of the heavy path which contains $a \in V(T)$.
So to store the heavy paths to which all nodes of $T$ belong we need $O(n \log n)$ bits.
The following function is supported by $F$.
\begin{itemize}
    \item $\texttt{getHeavyPath}(a)$: Returns the heavy path number of $H \in \mathcal{H}$ to which $a \in V(T)$ belongs in constant time. 
\end{itemize}

\noindent
{\bf Array, $L$.} This is a one dimensional array of length $|\mathcal{H}|$. $L[i]=l$ where $l$  is the level in  $\mathcal{T}$  to which heavy path $H_i \in \mathcal{H}$ belongs. Each entry in the array  uses $O(\log \log n)$ bits, since from Lemma~\ref{lem:hpd_prop}, there are $O(\log n)$ levels in $\mathcal{T}$. Since $|\mathcal{H}| \le n$, the total space taken by $L$ is $O(n \log \log n)$ bits. The following function is supported.
\begin{itemize}
    \item $\texttt{getLevel}(j)$: Returns in constant time the level to which heavy path $H_j \in \mathcal{H}$, for $j \in [n]$, belongs in $\mathcal{T}$.
\end{itemize}

\noindent
{\bf Array, $E$.} This is a $K \times 2n$ two dimensional array. Each row corresponds to a level of $\mathcal{T}$ and column corresponds to a vertex in $U_l, 1 \le l \le K$; by Lemma~\ref{lem:2nvertices}, $U_l$ has at most $2n$ vertices. $E[l][v]=i$ where $i$ is the index of the path $P_i \in \mathcal{P}$ that has non-empty  intersection  with  a heavy path $H$ at  level $l$ and $v$ is the vertex in $G_l$ corresponding to $P_i \cap H$. For vertex labels that are not present in $U_l$, $E[l][v]$ stores 0. The total space taken by $E$ is $O(n \log^2 n)$ as there are $2n \log n$ entries and each entry takes $\log n$ bits.

\begin{itemize}
    \item $\texttt{getPathIndex}(l,v)$: Returns  the path index stored at $E[l][v]$ in constant time given $1 \le l \le K$ and $1 \le v \le 2n$ as input. 
\end{itemize}

\noindent
{\bf The Path Intersection Table (PIT).} $PIT$ is an $n \times K$ two dimensional array of records with rows corresponding to paths of $\mathcal{P}$ and columns to the levels of $\mathcal{T}$. For path index $i$ and level $l$, each record consists of a bit and two vertex labels $v,v' \in V(G_l)$ such that $v$ denotes $P_i \cap  H_w$ and $v'$ denotes $P_i \cap H_{w'}$ for some heavy paths $H_w$ and $H_{w'}$ corresponding to nodes $w,w' \in V(\mathcal{T})$ at level $l$. The entries in $PIT$ are as follows:
\begin{enumerate}
    \itemsep0em 
    \item $PIT[i][l]=1$, if the path $P_i$ has non-empty intersection with some heavy path at level $l$ else $PIT[i][l]=0$. Storing this information takes $n \log n$ bits.
    \item $PIT[i][l]$ stores the labels of the two vertices in interval graphs $G_{H_w}$ and  $G_{H_{w'}}$ corresponding to $P_i \cap H_w$ and $P_i \cap H_{w'}$ where $w,w' \in V(\mathcal{T}$ at level $l$. If $P_i$ does not belong to the level $l$ then we store $\texttt{NULL}$.  If $P_i$ belongs to the level $l$ but to only one interval graph, say $G_{H_w}$, then $PIT[i][l]$ stores the label of the vertex in $G_{H_w}$. Each entry of the PIT takes at most $2 \log n$ bits, since by  Lemma~\ref{lem:2nvertices} there are at most two labels for a path per level.
\end{enumerate}
$PIT$ has $n \log  n$ entries and each entry takes $O(\log n)$ bits. Thus, total space needed is $O(n  \log^2 n)$. PIT is constructed as follows. Note that as per Proposition~\ref{prop:consecu}, each heavy path in $\mathcal{H}$ is an interval. 
\begin{enumerate}
    \itemsep0em 
    \item Initialize an array of counters $c[l] \leftarrow 0$ for each level $1 \le l \le K$.
    \item For each  $i \in [n]$ perform the following steps. 
    \begin{enumerate}
        \itemsep0em 
        \item For each  $a \in V(P_i)$ and the heavy path $j \leftarrow \texttt{getHeavyPath}(a)$ such that $a$ is the first vertex of $P_i$ that is also in $H_j$, do the following steps.
        \begin{enumerate}
            \itemsep0em 
            \item $l \leftarrow  \texttt{getLevel}(j)$.
            \item Store new vertex label $c[l] \leftarrow c[l]+1$ corresponding to path $P_i \cap H_j$ at level $l$ at $PIT[i][l]$. 
        \end{enumerate}
    \end{enumerate}
\end{enumerate}

\noindent
The following functions are supported.
\begin{enumerate}
    \itemsep0em 
    \item $\texttt{isPresent}(i,l)$: Returns true if $PIT[i][l]=1$ in constant time for path index $i \in [n]$ and the level $1 \le l \le K$.
    \item $\texttt{getVertices}(i,l)$: Returns the vertex labels stored at $PIT[i][l]$ in constant time for path index $i \in [n]$ and level $1 \le l \le  K$. The vertices are ordered based on the total order of the heavy paths that define them.
\end{enumerate}


\noindent
{\bf The Interval Graph Table (IT).} $IT$ is a one dimensional  array of length $K$. As a consequence of Lemma~\ref{lem:hpig}, $U_l,  1  \le l \le K,$ is an interval graph. For level $l$, $IT[l]$ stores $U_l$ using the method of Acan et al.~\cite{HSSS}. Thus, the total space taken by IT is $O(n \log^2 n)$. IT can be constructed in polynomial time as follows. Populate $IT[l]$, for each level $l$ of $\mathcal{T}$, using the following steps.

\begin{enumerate}
    \itemsep0em 
    \item Let $S_l$ be the set of heavy paths at level $l$. Obtain $S_l$  from $L$ and sort it in non-decreasing order. 
    \item Let \texttt{UNUSED}=0 and \texttt{USED}=1. For each $w \in S_l$ and path $P_i \in \mathcal{P}, i \in [n],$ do the following after initializing the bit vector $B$ of length $2n$ to $\texttt{UNUSED}$.
        \begin{enumerate}
            \itemsep0em 
            \item If $P_i \cap H_w \ne \phi$ then add the vertex returned by $\texttt{getVertices}(i,l)$ that is marked \texttt{UNUSED} in $B$ to $V(G_{H_w})$. Once a vertex corresponding to a path in $PIT$ is added to the interval graph it is marked as \texttt{USED} in $B$.
            \item For every vertex $u$ added to $V(G_{H_w})$, add $u$ into a temporary array $\texttt{TEMP}[w]$ along with $V(P_i \cap H_w)$.
        \end{enumerate}
    \item For every $w \in S_l$, add edges to interval graph $G_{H_w}$ as follows. 
    \begin{enumerate}
        \itemsep0em
        \item For all pairs of vertex labels $u$ and $v$ in $\texttt{TEMP}[w]$ where $u$ corresponds to $P_i \cap H_w$ and $v$ corresponds to $P_j \cap H_w$, add edge $\{u,v\}$ to $E(G_{H_w})$ if $V(P_i \cap H_w) \cap V(P_j \cap H_w) \ne \phi$.
    \end{enumerate}
     
    \item Finally, we get $U_l=\{G_{H_w} \mid w \in S_l\}$.
\end{enumerate}

\noindent
$U_l$ thus obtained can  now be stored  using the data  structure of~\cite{HSSS}. The following functions are supported by IT.
\begin{enumerate}
    \itemsep0em 
    \item $\texttt{adjacentIG}(u,v,l)$: Returns true if $u,v \in V(U_l)$ are adjacent in constant time.  This adjacency check is delegated to  the interval graph representation of~\cite{HSSS}. \cite{HSSS} supports constant time adjacency query.
    \item $\texttt{neighbourhoodIG}(u,l)$: Returns the neighbours of vertex $u \in V(U_l)$ in $O(d_u)$ time where $d_u$ is the  degree of vertex $u$. The query is delegated to  the interval graph representation of~\cite{HSSS}. \cite{HSSS} returns neighbours in constant time per neighbour.
\end{enumerate}


\noindent
{\bf Array $R$.} This is an $n \times 2$ two dimensional array. For $P_i \in \mathcal{P}$, $R[i][1]=a_i$ and $R[i][2]=b_i$ where $a_i$ and $b_i$ are the lowest and highest levels to  which heavy paths $H_w, H_{w'} \in \mathcal{H}$ belong in  $\mathcal{T}$ such that $H_w \cap P_i \ne \phi$ and $H_{w'} \cap P_i \ne \phi$. Since $P_i$ is  a path, it  has  non-empty  intersection  with some heavy path at all  the levels in the range $[a_i,b_i]$. We say, $P_i$ spans the levels from $a_i$ to $b_i$ and denote this range by  an  interval $I_i=[a_i,b_i], 1 \le a_i \le b_i \le K$. Each row of $R$ consists of two values, each taking $O(\log \log n)$ bits since $K \le \log n$. $R$ takes a total space of $O(n \log \log n)$ bits. \\
The following functions are supported by $R$:
\begin{enumerate}
    \itemsep0em 
    \item $\texttt{getEndPoints}(i)$: Returns the end points of $I_i$ for path $P_i \in \mathcal{P}$  in constant time for path index $i \in [n]$.
    \item $\texttt{getMinLevel}(i,j)$: Returns the left end point of $I_i \cap I_j$ in constant time for path indices $i,j \in [n]$ if $I_i \cap I_j \ne \phi$ else returns 0. The function returns:
    \begin{enumerate}
        \itemsep0em 
        \item if $b_i < a_j$ or $b_j < a_i$ then  0
        \item else if $a_i \le b_j$ then $a_i$ 
        \item else if $a_j \le b_i$ then $a_j$
    \end{enumerate} 
\end{enumerate}

\noindent
{\bf Array $A$.} This is a one dimensional array of length $M$. $A[a],  1\le a \le M,$ stores the list of paths that have their lca at node $a \in  V(T)$. A path have only  one lca and it takes $\log n$ bits to store this information as $M \le n$. For $n$ paths it takes $O(n \log n)$ bits. The following function is supported.
\begin{itemize}
    \item $\texttt{getPathsLCA}(a)$: Returns paths in $\mathcal{P}$ with lca at node $a \in V(T)$ in constant time.
\end{itemize}

\noindent
{\bf Heavy path tree $\mathcal{T}$.} The heavy path tree of clique tree $T$ is stored using the method of Lemma~\ref{lem:2ntree} in $\mathcal{T}$. Since $T$ is an ordinal tree, $\mathcal{T}$ is also ordinal. $\mathcal{T}$ takes $2n + o(n)$ bits and supports all the methods of ordinal trees as supported by the data structure of Lemma~\ref{lem:2ntree}.

\noindent
{\bf Array $H$.} This  is a  one dimensional array of length $|\mathcal{H}|$; see Figure~\ref{fig:arrayH} for an example. Let $w \in V(\mathcal{T})$ have $n_w$ children. Contents of $H$  are as follows.
\begin{itemize}
    \itemsep0em
    \item $H[w]$ stores a one dimensional array $C$ of length $n_w$ with a location for each of the children of $w$.
    \item $H[w]=\texttt{NULL}$ if $w$ does not have any children.
\end{itemize}
Since $\mathcal{T}$ is an  ordinal tree, the children of a  node are  ordered. For the $i-$th child of $w$, denoted $c$, with $n_c$ children, $C[i]$ stores a one dimensional array $D$ of length $n_c+1$. Contents of $D$  are as follows.
\begin{itemize}
    \itemsep0em
    \item $D[j][1]=d, 1 \le j \le n_c,$ where $d$ is the $j-$th child of $c$ and $D[j][2]$ contains a list that stores the paths that contain edges $\{d,c\}$ and $\{c,w\}$ where $w,c,d$ belong to consecutive levels $l_1 < l_2 < l_3$, respectively, in $\mathcal{T}$.
    \item $D[n_c+1][1]=\texttt{NULL}$ and $D[n_c+1][2]$ stores the list of paths that contain only $\{c,w\}$ and no light edge incident on $c$ in the sub-tree rooted at $c$.
    \item If $c$  does not have a child then $D[1][1]=\texttt{NULL}$  and $D[1][2]$ contains the list of paths that contain light  edge $\{c,w\}$. 
\end{itemize}
$C$ and $D$ are of size $O(n \log n)$ bits as they store entries for edges of $\mathcal{T}$ which, as a consequence of Remark~\ref{rem:maxn}, is at most $n-1$. Thus, $H$, $C$, and $D$ take a total of $O(n \log n)$ bits. The following function is supported.
\begin{itemize}
    \item $\texttt{getDistinctPaths}(w_1,w_2,w_3)$: Returns,  in constant time,  the  list of paths  that contain light edge $\{w_1,w_2\}$ but not $\{w_2,w_3\}$ where $\{w_1,w_2\},\{w_2,w_3\} \in E(\mathcal{T})$ and $w_1, w_2, w_3$ lie on consecutive levels $l_1<l_2<l_3$, respectively, in $\mathcal{T}$.
\end{itemize}

\begin{figure}[ht]
\centering
\includegraphics[width=1\textwidth]{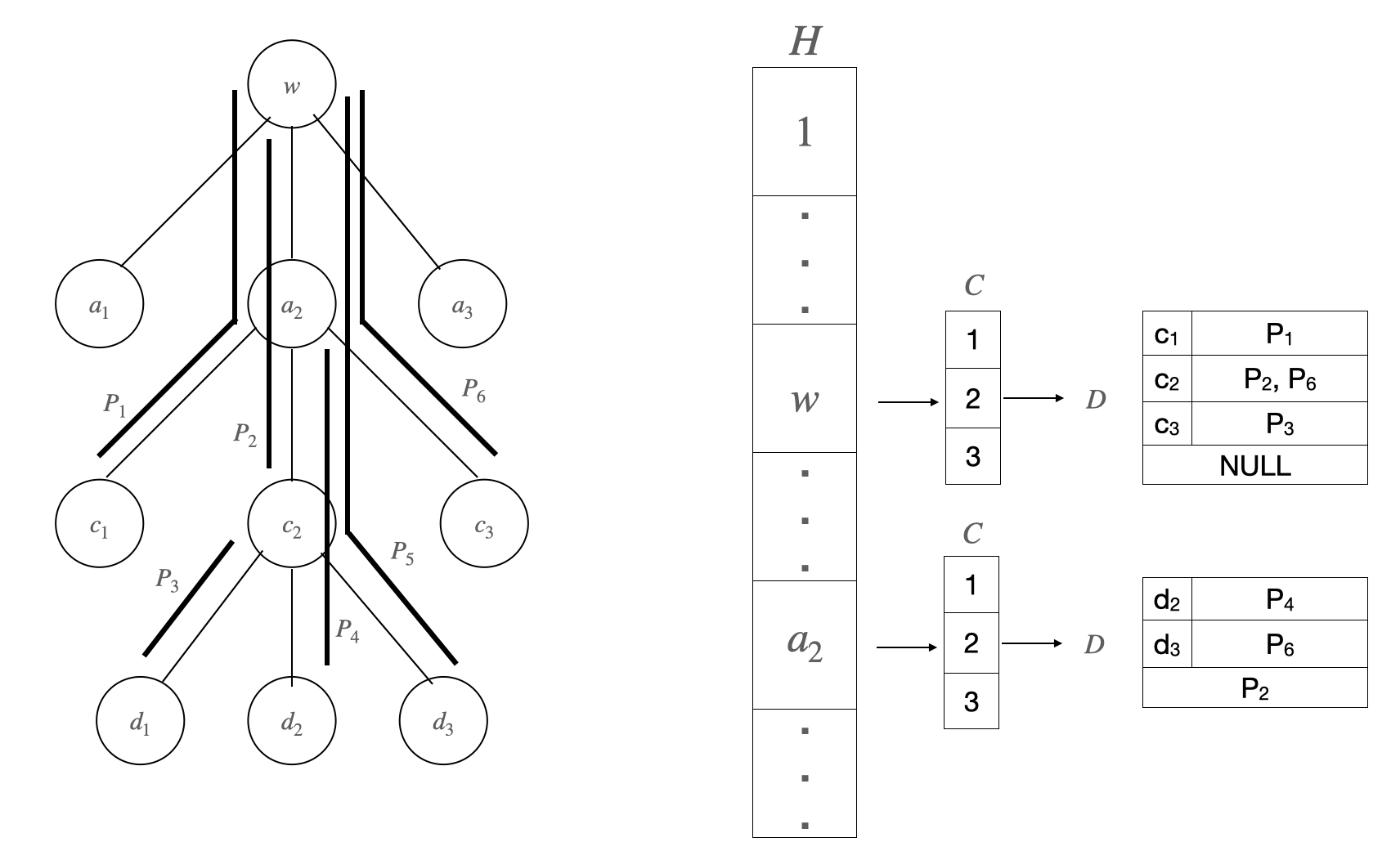}
\caption{A part of an example heavy path tree $\mathcal{T}$, is shown on the left side and array $H$ is shown on the right side. $H[w]$ contains $C$ with three entries corresponding to its children $\{a_1,a_2,a_3\}$. The array $D$ corresponding to  child $a_2$ at $C[2]$ is also shown. The first entry in the list stored at $D$ corresponds to $c_1$ and is associated with a list containing only one entry, $P_1$. This means $P_1$ contains light edge $\{c_1,a_2\}$ and $\{a_2,w\}$. The last entry in $D$ is $\texttt{NULL}$ which implies there is no path that starts at $a_2$ and contains light edge $\{a_2,w\}$. Notice that the $D$ corresponding to child $c_2$ of $a_2$, contains $P_2$ as the last entry.}
\label{fig:arrayH}
\end{figure}

\begin{lemma}
Let $\mathcal{T}$ be the ordinal heavy path tree and  $\{w_1,w_2\},\{w_2,w_3\} \in E(\mathcal{T})$ be two light edges such that $w_1, w_2, w_3$ lie on consecutive levels $l_1<l_2<l_3$, respectively, in $\mathcal{T}$. There exists a function that returns the list of paths that contain $\{w_1,w_2\}$ but not $\{w_2,w_3\}$ in constant time per path returned.    
\end{lemma}
\begin{proof}
The function $\texttt{getDistinctPaths}(w_1,w_2,w_3)$ is implemented as follows.
\begin{enumerate}
    \item $r \leftarrow \texttt{child\_rank}(w_2)$. $\texttt{child\_rank}$ is a function supported by ordinal tree $\mathcal{T}$ that returns the number of siblings to the left of $w_2$. It takes constant time  as per Lemma~\ref{lem:2ntree}.
    \item Obtain array $D$ from $C[r+1]$ stored in $H[w_1]$. Let $L'$ denote the list obtained by concatenating the lists stored at $D$ except the list corresponding to $w_3$.
    \item Return $L'$.
\end{enumerate}
$\texttt{child\_rank}$ takes constant time as per Lemma~\ref{lem:2ntree}. Concatenating each list into one takes constant time per list concatenated.  As each list contains at least one neighbour, the time taken is $O(1)$ per path returned.
\end{proof}

\begin{lemma}
\label{lem:sedspace}
There exists an $O(n \log^2 n)$-bit space-efficient data structure for path graphs.
\end{lemma}
\begin{proof}
The space taken by the components of the space-efficient data structure for path graphs are as follows:
\begin{enumerate}
    \itemsep0em 
    \item Array $F$ that contains the heavy paths to which each node of $T$ belongs takes $O(n \log n)$ bits.
    \item Array $L$ stores the level to  which each heavy path belongs taking $O(n \log \log n)$  bits. $R$ stores the ranges of levels in ${\cal T}$ that a path spans taking $O(n \log \log n)$ bits.
    \item For every level $l$, the path index corresponding to each of the vertex labels in $U_l$ is stored in array $E$ using $O(n \log^2 n)$ bits. $PIT$ stores the levels to which paths in $\mathcal{P}$ belong. For each level $l$, the vertex labels in $U_l$ corresponding to a path in $\mathcal{P}$ is stored using $O(n \log^2 n)$ bits. For each level $l$, $IT$ stores the interval graph $U_l$  taking $O(n \log^2 n)$ bits using the representation of~\cite{HSSS}.
\end{enumerate}
Thus, the entire space-efficient data structure uses $O(n \log^2 n)$ bits. 
\end{proof}

\subsection{Adjacency and Neighbourhood Queries}
Next, we present the algorithms for the adjacency, neighborhood and degree queries and their time complexities. We have the following useful lemmata that we will use in the implementation of the queries.

\begin{lemma}
\label{lem:lo_olp}
Consider paths with indices $i,j \in [n]$ such that $[l_1,l_2]$ is the maximal range of levels with $PIT[i][l]=PIT[j][l]=1$ for all $l \in [l_1,l_2]$. If paths $P_i$ and $P_j$ do not intersect in $U_{l_1}$ then they do not intersect at any level  $l > l_1$.
\end{lemma}
\begin{proof}
If paths $P_i,P_j \in \mathcal{P}$ do not intersect in $U_{l_1}$ then there are two possibilities:
\begin{enumerate}
    \itemsep0em 
    \item They intersect two different heavy paths at level $l_1$ in the heavy path tree. In this case, they will not intersect in any  level greater than $l_1$ as they are contained in two different branches of the heavy path tree.
    \item They intersect the same heavy path  at level $l_1$ but different heavy paths at  levels greater than $l_1$. Thus, at any level $l > l_1$ they are in different branches of the heavy path tree and so will not  intersect. 
\end{enumerate}
Hence, the lemma.
\end{proof}

\begin{lemma}
\label{lem:nocommonle}
Let $P=(l,r)$ and $Q=(s,t)$ be two paths in $\mathcal{P}$ with sequence of heavy sub-paths $\Pi_P$ and $\Pi_Q$, respectively. Also, let $V(P) \cap V(Q) \ne \phi$ such that there does not exist a light  edge $e$ such that $e \in E(P) \cap E(Q)$. The following are true.
\begin{enumerate}
    \itemsep0em
    \item There exists exactly one $\pi \in \Pi_P, \pi' \in \Pi_Q$ and heavy path $H=(h_1,h_2)$ such that $V(\pi) \cap V(\pi') \cap V(H) \neq \phi$.
    \item Further, either $E(P) \subseteq E(T_{h_1})$ or $E(Q) \subseteq E(T_{h_1})$ where $T_{h_1}$ is the sub-tree rooted at $h_1$.
    \item The lowest numbered node of $\pi$ is either the ${\normalfont \texttt{lca}(l,r)}$ or it is $h_1$ such that light edge $\{h_1,{\normalfont \texttt{parent}(h_1)}\} \in E(P)$.
    \item Either, ${\normalfont \texttt{lca}}(l,r) \in V(\pi)$ or ${\normalfont \texttt{lca}}(s,t) \in V(\pi)$.
\end{enumerate} 
\end{lemma}
\begin{proof}
The proof is as follows:
\begin{enumerate}
    \itemsep0em
    \item  $V(P) \cap V(Q)$ is contained in some $H \in \mathcal{H}$, since by Lemma~\ref{lem:vparti}, the nodes of $T$ are  partitioned among the heavy  paths. There is exactly one such $H$, as $V(P) \cap V(Q)$ does not have pair of nodes $u,v$ such that $\{u,v\}$ is a light edge in $T$.
    \item We consider two cases here.
    \begin{enumerate}
        \itemsep0em
        \item $h_1$ is the root of $T$: In this case, trivially $E(P) \subseteq E(T_{h_1})$ and $E(Q) \subseteq E(T_{h_1})$ since $T_{h_1}=T$.
        \item $h_1$ is not the root of $T$: If both $P$ and $Q$ do not contain light edge $\{\texttt{parent}(h_1),h_1\}$, then $E(P) \subseteq E(T_{h_1})$ and $E(Q) \subseteq E(T_{h_1})$. Else, since $P$ and $Q$ do  not share a light edge, either $\{\texttt{parent}(h_1),h_1\} \in E(P)$ or $\{\texttt{parent}(h_1),h_1\} \in E(Q)$. Without loss of generality, let it be an element of $E(P)$. Then, $E(Q)  \subseteq E(T_{h_1})$.  Thus, if $P$ and $Q$ do not share  a  light edge, at least one of the paths must be contained inside $T_{h_1}$. 
    \end{enumerate}
    \item Based on the earlier proved statement, we have two cases:
    \begin{enumerate}
        \itemsep0em
        \item $E(P) \subseteq E(T_{h_1})$: In  this case, $\texttt{lca}(l,r) \in V(\pi)$ and is the lowest numbered node in $\pi$.
        \item $E(P) \nsubseteq E(T_{h_1})$: In  this case, $h_1 \in V(\pi)$  and is the lowest numbered node in $\pi$.
    \end{enumerate}
    \item There are two possibilities based on the lowest numbered node in $\pi$. 
    \begin{enumerate}
        \itemsep0em
        \item If the lowest numbered vertex of $\pi$ is the $\texttt{lca}(l,r)$ then the statement follows trivially.
        \item If the lowest numbered vertex of $\pi$ is $h_1$ such that light edge $\{h_1,\texttt{parent}(h_1)\} \in E(P)$ then $\texttt{lca}(s,t) \in V(\pi) \cap V(\pi')$; since $E(Q) \subseteq E(T_{h_1})$. Hence, the result.
    \end{enumerate}
    

\end{enumerate}
\end{proof}

\begin{lemma}
\label{lem:onepathpermaxclique}
For every $a \in  V(T)$ there  exists $P=(l,r)$ in $\mathcal{P}$ such  that $a={\normalfont \texttt{lca}}(l,r)$.
\end{lemma}
\begin{proof}
Every $a \in V(T)$ corresponds to  a maximal clique of $G$. We categorise maximal cliques of $G$ in the following manner.
\begin{enumerate}
    \itemsep0em
    \item Maximal clique $C \in \mathcal{C}$ contains a simplicial vertex $v \in V(G)$: Let $P_v=(a,a)$  be the path corresponding to  $v$ where  $a \in V(T)$ is the  node corresponding to $C$. Then, $\texttt{lca}(a,a)=a$ and the statment follows.
    \item Maximal clique $C \in \mathcal{C}$ does not contain a simplicial vertex: Since $C$ is a maximal clique, $V(C) \nsubseteq V(C')$ for $C' \in \mathcal{C}$ and  $C \neq C'$. Let $a \in V(T)$  be the node corresponding to $C$. If all the vertices of $C$ correspond to paths that contain $\texttt{parent}(a)$ then $C \subseteq  C'$ where $C'$ is the maximal clique corresponding to $\texttt{parent}(a)$. Thus, at  least one of the following must be true:
    \begin{enumerate}
        \itemsep0em
        \item there is a path containing $a$ that starts at a descendant  of $a$ and ends at another descendant of $a$, or
        \item there is a path that starts at $a$ and ends at a descendant of $a$.
    \end{enumerate}
    Let that path be $P=(l,r)$. Then, $\texttt{lca}(l,r)=a$.
    
\end{enumerate}
\end{proof}

\noindent 
{\bf Adjacency query.} The adjacency query of Algorithm~\ref{alg:sedadjacency} takes the index $i,j$ of paths $P,Q \in \mathcal{P}$ and the space-efficient representation constructed in Section~\ref{sec:sedconstruction} as input and checks if $P_i,P_j \in \mathcal{P}$ have  a non-empty intersection. 

\begin{algorithm}[ht!]
\label{alg:sedadjacency}
\caption{For path graph $(T,\mathcal{P})$ and two paths $P_i, P_j \in \mathcal{P}$, the function checks if $V(P_i) \cap V(P_j) \neq \phi$.}
\DontPrintSemicolon
    \SetKwFunction{Fadjacency}{adjacency}

    \SetKwProg{Fn}{Function}{:}{}
    \Fn{\Fadjacency{$i,j$}}{
        $l={\normalfont \texttt{getMinLevel}}(i,j)$\;
        \If{$l \neq 0$}{
            $\{u_1,u_2\} \leftarrow {\normalfont \texttt{getVertices}}(i,l)$\;
            $\{v_1,v_2\} \leftarrow {\normalfont \texttt{getVertices}}(j,l)$\;
            \If{for any pair $\{a,b\} \in  \{\{u_1,v_1\},\{u_1,v_2\},\{u_2,v_1\},\{u_2,v_2\}\}$ ${\normalfont \texttt{adjacentIG}}(a,b,l)$ is true}{
                \KwRet true\;
            }
        }
        \KwRet false\;
    }    
\end{algorithm} 

\begin{lemma}
\label{lem:adjsed}
Given path indices $i,j \in [n]$ and the space-efficient representation as input, the function ${\normalfont\texttt{adjacency}}(i,j)$ checks if paths corresponding to $i$ and $j$ have  a non-empty intersection in constant time.    
\end{lemma}
\begin{proof}
Due to Lemma~\ref{lem:lo_olp} it is only required to check if $P_i$ and $P_j$ intersect in level $\texttt{getMinLevel}(i,j)$. If $\texttt{getMinLevel}(i,j) \ne 0$ then in Line 6 of Algorithm~\ref{alg:sedadjacency} we check if any one of the four pairs of vertex labels paths $P_i$ and $P_j$ in  interval graph $U_l$ are adjacent. The  vertex labels for paths $P_i$ and $P_j$ in $U_l$ are obtained using the function $\texttt{getVertices}(i,l)$ and $\texttt{getVertices}(j,l)$, respectively in Lines 4 and 5.  The adjacency check in the interval graph is done using the function $\texttt{adjacenctIG}$ in Line 6. Since $\texttt{getMinLevel}$, $\texttt{getVertices}$  and  $\texttt{adjacenctIG}$ are constant time functions  adjacency check can be completed in constant time.    
\end{proof}

\noindent
{\bf Neighbourhood query.} The neighbourhood  query can be implemented as  shown in Algorithm~\ref{alg:sedneighbourhood2}. It takes the path index and the space-efficient representation constructed in Section~\ref{sec:sedconstruction} as input and lists all  the paths that have a non-empty intersection with the input path. 


\begin{algorithm}[ht!]
\label{alg:sedneighbourhood2}
\caption{For space-efficient representation of path graph $(T,\mathcal{P})$ and an input path $P_i \in \mathcal{P}$, the function returns its neighbours.}
\DontPrintSemicolon
    \SetKwFunction{Fneighbourhood}{neighbourhood}

    \SetKwProg{Fn}{Function}{:}{}
    \Fn{\Fneighbourhood{$i$}}{
        Set $N_i,E_1,E_2$ to $\texttt{NULL}$\;
        $[l,r] \leftarrow \texttt{getEndPoints}(i)$ \;
        $p \leftarrow \texttt{lca}(l,r)$\;
        \While{$l \neq p$}{
            Add $\texttt{getPathsLCA(l)}$ to $N_i$ \;
            $c \leftarrow \texttt{parent}(l)$ \;
            \If{${\normalfont \texttt{getHeavyPath}}(l)\ne {\normalfont \texttt{getHeavyPath}}(c)$}{
                Add $\{c,l\}$ to end of $E_1$\;
            }
            $l \leftarrow c$\;
        }
        $h \leftarrow \texttt{getHeavyPath}(l)$\;
        $L \leftarrow \texttt{getLevel}(h)$\;
        $\{v_1,v_2\} \leftarrow \texttt{getVertices}(i,L)$\;
        Add $\texttt{neighbourhoodIG}(v_1,L)$ to $N_i$ \;
        \While{$r \neq p$}{
            Add $\texttt{getPathsLCA(r)}$ to $N_i$ \;
            $c \leftarrow \texttt{parent}(r)$ \;
            \If{${\normalfont \texttt{getHeavyPath}}(r)\ne {\normalfont\texttt{getHeavyPath}}(c)$}{
                Add $\{c,r\}$ to end of  $E_2$\;
            }
            $r \leftarrow c$\;
        }
        Concatenate $E_2$ to the end of $E_1$ and assign it to $E$\;
        $e' \leftarrow \texttt{NULL}$\;
        \ForEach{$e=(w^{e}_1,w^{e}_2)$ in $E$}{
            \If{$e' = {\normalfont \texttt{NULL}}$}{
                Add $\texttt{getDistinctPaths}(w_1^{e},w_2^{e},{\normalfont \texttt{NULL}})$ to $N_i$\;
            }
            \Else{
                Add $\texttt{getDistinctPaths}(w_1^{e},w_2^{e},w_2^{e'})$ to $N_i$\;
            }
            $e' \leftarrow e$\;
        }
    }    
\end{algorithm} 

\begin{lemma}
\label{lem:sedneighbourhood}
Given the space-efficient data structure for $G$ and the index $i$  of path $P \in \mathcal{P}$ as input, ${\normalfont \texttt{neighbourhood}}(i)$ returns the neighbours of $P$ in $O(d)$ time where $d$ is the degree of $P$.   
\end{lemma}
\begin{proof}
We will prove that Algorithm~\ref{alg:sedneighbourhood2} enumerates neighbours of $P$ at least once and at most a constant number of times. It follows from Lemma~\ref{lem:charinter} that intersecting paths are of two types, namely, ones with no common light edge and ones with at least one light edge. We have the following cases.
\begin{enumerate}
    \item Neighbours that share no light edge with $P$: Let $\Pi$ be the set of heavy sub-paths of $P$. From Lemma~\ref{lem:nocommonle}, neighbours with no common edges with $P$ are characterised by $\texttt{lca}(s,t) \in V(\pi)$ and/or $\texttt{lca}(l,r) \in V(\pi)$ where $\pi \in  \Pi$. In Lines 5 to  10 and Lines 15 to 20 of Algorithm~\ref{alg:sedneighbourhood2}, the paths with lca in any node $u \in V(\pi)$ are added to $N_i$ using the function $\texttt{getPathsLCA}$. Further, in Lines 11 to 14, neighbours of $P$, for instance,  $Q$ such  that $V(P) \cap V(Q) \cap V(H) \neq \phi$ such that $H \in \mathcal{H}$ and $\texttt{lca}(l,r) \in V(H)$, are added to $N_i$. Thus, neighbours that share no light edge with $P$ will be counted at least once. A path $Q$  that has lca in $V(\pi)$ for a $\pi \in \Pi$ such that $\texttt{lca}(l,r) \in V(\pi)$ will be counted at most twice. 
    \item Neighbours that share at least one light edge with $P$: Let $p=\texttt{lca}(l,r)$. In Lines 23 to 28 of Algorithm~\ref{alg:sedneighbourhood2}, the light edges that are encountered as we traverse from $l$ to $p$ and  $r$  to $p$, respectively, are considered. $\texttt{getDistinctPaths}$ is used to add paths that contain these light edges to $N_i$. $\texttt{getDistinctPaths}$ do not repeat paths that are counted on light edges already visited as $P$ is traversed from $l$ to $p$.  Also, $\texttt{getDistinctPaths}$ do not repeat paths that are counted on light edges already visited as $P$ is traversed from $r$ to $p$. Thus, every neighbour sharing a light edge with $P$ is counted exactly once. 
\end{enumerate}
Some neighbours of $P$ can share a light edge with it and also satisfy, for some $\pi \in \Pi$, $\texttt{lca}(s,t) \in V(\pi)$ or $\texttt{lca}(l,r) \in V(\pi)$. In this case too, they will be over-counted at most a constant number of times. \\
Functions $\texttt{getEndPoints}$, $\texttt{lca}$, $\texttt{getPathsLCA}$, $\texttt{getHeavyPath}$, $\texttt{getLevel}$, $\texttt{neighbourhoodIG}$, and $\texttt{getDistinctPaths}$ are constant time functions. Loops at Line 5 and 15 repeat a maximum of $O(d)$ times since as per Lemma~\ref{lem:onepathpermaxclique}, each node in $V(P)$ is a maximal clique that contributes at least one distinct neighbour. By the same argument, the loop at Line 23 repeats $O(d)$ times as the number of edges in $P$ is $O(d)$. Hence, the time complexity of the neighbourhood query is $O(d)$.
\end{proof}

\noindent
{\bf Degree query.} The degree of each vertex can be stored using $n \log n$ bits and the degree query can be solved in constant time.

\noindent
\begin{proof}[Proof of Theorem~\ref{thm:nlog2n}]
Lemma~\ref{lem:sedspace} shows that there exists an $O(n \log^2 n)$ bit space-efficient data structure for path graphs. Given this representation as input Lemma~\ref{lem:adjsed} shows that adjacency between  vertices can be checked in constant time. Similarly, using this representation, Lemma~\ref{lem:sedneighbourhood} shows an $O(d)$ neighbourhood query. Also, degree query is satisfied in constant time by accessing it from an array. Thus, we conclude Theorem~\ref{thm:nlog2n}.
\end{proof}

\section{Conclusion} 
In this work, we designed efficient data structures for path graphs. In the future, we believe some of the following directions would be interesting to explore regarding path graphs.
\begin{enumerate}
    \itemsep0em 
    \item The best implementations of BFS and DFS are of significant interest as many other problems for path graphs use them as subroutines. In the work by Acan et al.~\cite{HSSS}, for interval graphs we can see that the representation permits very efficient BFS and DFS algorithms. Can we perform BFS/DFS efficiently on path graphs assuming our representation?
    \item Can we show time/space trade-off lower bounds for our data structures? More specifically, can we prove tight space lower bound of redundancy with respect to query time? 
    \item Are there other graph classes amenable to our techniques for designing succinct data structures?
\end{enumerate}


\Section{Reference}
\bibliographystyle{IEEEbib}
\bibliography{refs}

\end{document}